\documentclass[11pt]{article}
\usepackage[margin=1in]{geometry}
\pdfoutput=1
\usepackage[english]{babel}

\let\coloneqq\relax

\usepackage{amsfonts}
\usepackage{amsthm}
\usepackage{amssymb}
\usepackage{amsmath}
\usepackage{bbold}
\usepackage{cite}
\usepackage{bbm}
\usepackage{nonfloat}
\usepackage[pdftex]{hyperref}
\usepackage[capitalize,nameinlink]{cleveref}
\usepackage{braket}
\usepackage{dsfont}
\usepackage{mathdots}
\usepackage{mathtools}
\usepackage{enumerate}
\usepackage[shortlabels]{enumitem}
\usepackage{stmaryrd}
\usepackage{graphicx}
\usepackage{stackengine}
\usepackage{scalerel}
\usepackage{xr}
\usepackage[dvipsnames]{xcolor}
\usepackage{xcolor}
\usepackage{array}
\usepackage{makecell}
\newcolumntype{x}[1]{>{\centering\arraybackslash}p{#1}}
\usepackage{tikz}
\usepackage{pgfplots}
\usetikzlibrary{shapes.geometric, shapes.misc, positioning, arrows, arrows.meta, decorations.pathreplacing, decorations.pathmorphing, patterns, angles, quotes, calc}
\usepackage{booktabs}
\usepackage{xfrac}
\usepackage{siunitx}
\usepackage{centernot}
\usepackage{comment}
\usepackage{chngcntr}
\usepackage{algorithm}
\usepackage[noend]{algpseudocode}

\makeatletter
\def\thmhead@plain#1#2#3{%
  \thmname{#1}\thmnumber{\@ifnotempty{#1}{ }\@upn{#2}}%
  \thmnote{ {\the\thm@notefont#3}}}
\let\thmhead\thmhead@plain
\makeatother

\newcommand{\bb}{\begin{equation}\begin{aligned}\hspace{0pt}}
\newcommand{\bbb}{\begin{equation*}\begin{aligned}}
\newcommand{\ee}{\end{aligned}\end{equation}}
\newcommand{\eee}{\end{aligned}\end{equation*}}
\newcommand*{\coloneqq}{\mathrel{\vcenter{\baselineskip0.5ex \lineskiplimit0pt \hbox{\scriptsize.}\hbox{\scriptsize.}}} =}

\newcommand{\id}{\mathds{1}}
\newcommand{\R}{\mathds{R}}
\newcommand{\N}{\mathds{N}}
\newcommand{\C}{\mathds{C}}

\newcommand{\norm}[1]{\lVert #1 \rVert}
\newcommand{\NORM}[1]{\left\lVert #1 \right\rVert}
\newcommand{\abs}[1]{\lvert #1 \rvert}
\newcommand{\Tr}{\mathrm{Tr}}

\DeclareMathAlphabet{\pazocal}{OMS}{zplm}{m}{n}

\DeclareMathOperator{\Id}{Id}

\newcommand{\lsmatrix}{\left(\begin{smallmatrix}}
\newcommand{\rsmatrix}{\end{smallmatrix}\right)}

\stackMath

\stackMath

\makeatletter
\newcommand*\rel@kern[1]{\kern#1\dimexpr\macc@kerna}
\newcommand*\widebar[1]{%
  \begingroup
  \def\mathaccent##1##2{%
    \rel@kern{0.8}%
    \overline{\rel@kern{-0.8}\macc@nucleus\rel@kern{0.2}}%
    \rel@kern{-0.2}%
  }%
  \macc@depth\@ne
  \let\math@bgroup\@empty \let\math@egroup\macc@set@skewchar
  \mathsurround\z@ \frozen@everymath{\mathgroup\macc@group\relax}%
  \macc@set@skewchar\relax
  \let\mathaccentV\macc@nested@a
  \macc@nested@a\relax111{#1}%
  \endgroup
}

\counterwithin*{equation}{part}
\counterwithin*{figure}{part}

\tikzset{meter/.append style={draw, inner sep=10, rectangle, font=\vphantom{A}, minimum width=30, line width=.8, path picture={\draw[black] ([shift={(.1,.3)}]path picture bounding box.south west) to[bend left=50] ([shift={(-.1,.3)}]path picture bounding box.south east);\draw[black,-latex] ([shift={(0,.1)}]path picture bounding box.south) -- ([shift={(.3,-.1)}]path picture bounding box.north);}}}
\tikzset{roundnode/.append style={circle, draw=black, fill=gray!20, thick, minimum size=10mm}}
\tikzset{squarenode/.style={rectangle, draw=black, fill=none, thick, minimum size=10mm}}

\definecolor{Blues5seq1}{RGB}{239,243,255}
\definecolor{Blues5seq2}{RGB}{189,215,231}
\definecolor{Blues5seq3}{RGB}{107,174,214}
\definecolor{Blues5seq4}{RGB}{49,130,189}
\definecolor{Blues5seq5}{RGB}{8,81,156}

\definecolor{Greens5seq1}{RGB}{237,248,233}
\definecolor{Greens5seq2}{RGB}{186,228,179}
\definecolor{Greens5seq3}{RGB}{116,196,118}
\definecolor{Greens5seq4}{RGB}{49,163,84}
\definecolor{Greens5seq5}{RGB}{0,109,44}

\definecolor{Reds5seq1}{RGB}{254,229,217}
\definecolor{Reds5seq2}{RGB}{252,174,145}
\definecolor{Reds5seq3}{RGB}{251,106,74}
\definecolor{Reds5seq4}{RGB}{222,45,38}
\definecolor{Reds5seq5}{RGB}{165,15,21}

\pgfplotsset{width=10cm,compat=1.9}
\usetikzlibrary{decorations.markings}

\newcommand{\RN}[1]{\textup{\uppercase\expandafter{\romannumeral#1}}}

\newcommand{\symspace}[2]{\mathrm{Sym}_{k}(\bb{C}^d)}
\newcommand{\symproj}[2]{P_{\mathrm{Sym}}^{(d,k)}}

\newtheorem{proposition}{Proposition}
\newtheorem{theorem}{Theorem}
\newtheorem{lemma}{Lemma}

\newtheorem{remark}{Remark}

\newtheorem{fact}{Fact}
\newtheorem{problem}{Problem}

\numberwithin{lemma}{section}
\numberwithin{proposition}{section}
\numberwithin{definition}{section}
\numberwithin{theorem}{section}
\numberwithin{remark}{section}
\numberwithin{equation}{section}
\numberwithin{fact}{section}

\hypersetup{
	colorlinks = true,
	linkcolor = blue!60!black,
	citecolor = [rgb]{0.13,0.55,0.13},
	urlcolor = blue!60!black}

\DeclareRobustCommand{\orcidicon}{%
	\begin{tikzpicture}
	\draw[lime, fill=lime] (0,0) 
	circle [radius=0.16] 
	node[white] {{\fontfamily{qag}\selectfont \tiny ID}};
	\draw[white, fill=white] (-0.0625,0.095) 
	circle [radius=0.007];
	\end{tikzpicture}
	\hspace{-2mm}
}
\foreach \x in {A, ..., Z}{%
	\expandafter\xdef\csname orcid\x\endcsname{\noexpand\href{https://orcid.org/\csname orcidauthor\x\endcsname}{\noexpand\orcidicon}}
}

\newcommand{\eps}{\varepsilon}

\newcommand{\squeezingparameter}{z}

\title{\textbf{Efficient learning of bosonic Gaussian unitaries}}
\author{
Marco Fanizza\thanks{Inria, Télécom Paris - LTCI, Institut Polytechnique de Paris. Email: \href{matilto:marco.fanizza@inria.fr}{\texttt{marco.fanizza@inria.fr}}} 
\quad Vishnu Iyer\thanks{Department of Computer Science, The University of Texas at Austin. Email: \href{mailto:vishnu.iyer@utexas.edu}{\texttt{vishnu.iyer@utexas.edu}}}
\quad Junseo Lee\thanks{Institute of Computer Technology, Seoul National University. Email: \href{mailto:junseolee@fas.harvard.edu}{\texttt{junseolee@fas.harvard.edu}}}
\quad Antonio A. Mele\thanks{Dahlem Center for Complex Quantum Systems, Freie Universitat Berlin. Email: \href{mailto:a.mele@fu-berlin.de}{\texttt{a.mele@fu-berlin.de}}}
\quad Francesco A. Mele\thanks{NEST, Scuola Normale Superiore and Istituto Nanoscienze. Email: \href{mailto:francesco.mele@sns.it}{\texttt{francesco.mele@sns.it}}}}

\date{\vspace{-4ex}}
\begin{document}
\maketitle

\begin{abstract}
Bosonic Gaussian unitaries are fundamental building blocks of central continuous-variable quantum technologies such as quantum-optic interferometry and bosonic error-correction schemes.
In this work, we present the first time-efficient algorithm for learning bosonic Gaussian unitaries with a rigorous analysis. Our algorithm produces an estimate of the unknown unitary that is accurate to small worst-case error, measured by the physically motivated energy-constrained diamond distance.
Its runtime and query complexity scale polynomially with the number of modes, the inverse target accuracy, and natural energy parameters quantifying the allowed input energy and the unitary’s output-energy growth. The protocol uses only experimentally friendly photonic resources—coherent and squeezed probes, passive linear optics, and heterodyne/homodyne detection. We then employ an efficient classical post-processing routine that leverages a \emph{symplectic regularization} step to project matrix estimates onto the symplectic group.
In the limit of unbounded input energy, our procedure attains arbitrarily high precision using only $2m+2$ queries, where $m$ is the number of modes. To our knowledge, this is the first provably efficient learning algorithm for a multiparameter family of continuous-variable unitaries.
\end{abstract}

\newpage
\tableofcontents

\newpage
\section{Introduction}

Quantum learning theory~\cite{anshu2024survey} investigates the fundamental question of how efficiently one can extract information from quantum systems under realistic resource constraints. In recent years, there has been tremendous progress in our understanding of tasks for learning quantum systems~\cite{chen2022exponential,chen2023does,huang2022quantum,chen2025information,anshu2020sample,huang2022provably,huang2020predicting,huang2024learning,chen2024predicting,chen2024optimal,chen2025efficient}.

In particular, in the discrete-variable (DV) setting, a wide range of problems---from quantum state tomography~\cite{o2016efficient,haah2017sample} and property testing~\cite{montanaro2013survey,o2015quantum,o2017efficient,bubeck2020entanglement,chen2022toward,chen2022tight} to channel learning~\cite{mohseni2008quantum,chung2021sample,haah2023query}, as well as many other related tasks---have been systematically studied. The surveys~\cite{arunachalam2017guest,anshu2024survey} provide thorough overviews of advancements in these areas.

In contrast, the continuous-variable (CV) regime has only recently begun to attract substantial attention within quantum learning theory. Continuous-variable systems~\cite{Braunstein-review,andersen2010continuous,BUCCO} are indispensable in both experimental and theoretical quantum information science, yet their infinite-dimensional structure poses unique mathematical and statistical challenges. Recent works on learning CV systems have begun to chart this territory~\cite{mele2024learning, bittel2024optimalestimatestracedistance, bittel2025energyindependenttomographygaussianstates, fanizza2024efficienthamiltonianstructuretrace, gandhari_precision_2023, becker_classical_2023, oh2024entanglementenabled, liu2025quantumlearningadvantagescalable, coroi2025exponentialadvantagecontinuousvariablequantum, fawzi2024optimalfidelityestimationbinary, wu2024efficient, mobus2023diss, upreti2024efficientquantumstateverification, zhao2025complexityquantumtomographygenuine,symplectity_rank}, with particular progress in quantum state tomography~\cite{mele2024learning, bittel2024optimalestimatestracedistance, bittel2025energyindependenttomographygaussianstates, fanizza2024efficienthamiltonianstructuretrace}. These results illustrate both the potential and the limitations of existing techniques, while also underscoring the need for new approaches adapted to the unbounded nature of CV Hilbert spaces.

Moving beyond the learning of states, a natural and pressing direction is the study of quantum process learning. Just as quantum state tomography provides a window into the structure of states, quantum process tomography~\cite{mohseni2008quantum,o2004quantum,scott2008optimizing,haah2023query} allows us to fully characterize the transformations that act upon them. As one might expect, learning arbitrary continuous-variable processes is extremely expensive~\cite{mele2024learning}, and attention must be restricted to subclasses with nice properties for computational efficiency to be attainable.

Among all processes, bosonic Gaussian unitaries stand out as the most fundamental~\cite{BUCCO}. They arise ubiquitously in quantum optics~\cite{PhysRevLett.121.160502, PhysRevLett.86.5870,q_sensing_cv,BUCCO} and interferometry~\cite{Schnabel2010,SGravi2} as passive linear transformations, such as beam splitters and phase shifters, which preserve photon number while redistributing modes, and also comprise active linear optics, i.e. transformations employing squeezing. Because these operations form the backbone of optical experiments, the ability to benchmark and learn them with statistical rigor is of central importance. 

The physical relevance of Gaussian unitaries is profound: they describe linear transformations of position and momentum operators, which lie at the heart of continuous-variable quantum mechanics. In quantum computing, Gaussian unitaries form the foundation of many applications~\cite{Gottesman2001,Mirrahimi_2014,Ofek_nature2016,error_corr_boson,Guillaud_2019}, e.g.~Gaussian boson sampling~\cite{hamilton2017gaussian,kruse2019detailed}, which is a leading candidate for demonstrating quantum advantage. Even small inaccuracies in modeling or calibrating these transformations can drastically change the output distribution, underscoring the importance of precise learning algorithms. In the context of fault tolerance, Gaussian transformations are indispensable in the Gottesman–Kitaev–Preskill (GKP) code~\cite{Gottesman2001}, where they enable error correction and manipulation of encoded logical information. More broadly, Gaussian unitaries pervade quantum communication protocols~\cite{Wolf2007,TGW,PLOB,Mele_2025,usenko2025continuousvariablequantumcommunication}, continuous-variable quantum cryptography~\cite{grosshans2002continuous,navascues2006optimality,ralph1999continuous}, quantum metrology~\cite{nichols2018multiparameter,fadel2024quantum,SGravi2,PhysRevLett.121.160502, PhysRevLett.86.5870,q_sensing_cv}, and analog quantum simulation~\cite{daley2022practical}. Their ubiquity makes it imperative to rigorously understand the statistical complexity of learning them.

The task of learning bosonic Gaussian \emph{states}, that is, states prepared by bosonic Gaussian unitaries, is fairly well studied. Recent work gives an algorithm to learn these states scaling polynomially in the number of modes and an energy constraint parameter~\cite{mele2024learning}. Subsequent work improved the polynomial dependence in the sample complexity of this learning algorithm~\cite{bittel2024optimalestimatestracedistance,fanizza2024efficienthamiltonianstructuretrace, bittel2025energyindependenttomographygaussianstates}. However, the closely related problem of developing a time-efficient algorithm to learn Gaussian unitaries remained wide open. In particular, while Gaussian states can now be learned under polynomial resources, the learnability of Gaussian \emph{unitaries} themselves, even with natural restrictions, remains largely unknown. In fact, essentially no theoretical complexity guarantees are currently established for this problem. This motivates the central question of our work: 
\begin{center}
    \emph{Can we design a computationally efficient algorithm for learning bosonic Gaussian unitaries \\ to small worst-case error measured by a physically motivated distance?}
\end{center}
In this work, we provide an affirmative answer to this question by designing algorithms with rigorous theoretical guarantees.

\subsection{Main results}
We establish that arbitrary multi-mode bosonic Gaussian unitaries can be efficiently learned, with performance guarantees given in terms of the \emph{energy-constrained diamond norm}. 
At a high level, the so-called energy-constrained diamond norm~\cite{VV-diamond,Shirokov2018,PLOB,EC-diamond}, denoted $\norm{\cdot}_{\diamond,\bar n}$, captures the maximal advantage in distinguishing two quantum channels when the input states are restricted to have mean photon number at most $\bar n$. More specifically, $\frac12\norm{\Phi_1-\Phi_2}_{\diamond,\bar n}$ is the optimal probability bias for correctly identifying an unknown channel, promised to be either $\Phi_1$ or $\Phi_2$ with equal prior probability, using a single query and only input states with mean photon number at the input of at most $\bar n$. This metric is the physically relevant analogue of the diamond norm~\cite{Aharonov1998,Sumeet_book} for continuous-variable systems, and provides a worst-case guarantee on the performance of any tomography or learning protocol subject to energy limitations. In fact, without such an energy constraint, the standard (unconstrained) diamond norm exhibits undesirable properties for continuous-variable quantum channels, causing it to lose the strong physical meaning it possesses in finite-dimensional systems~\cite{VV-diamond}. A formal definition of energy-constrained diamond norm is given in \cref{sec:energy-constrained-diamond-distance}.

We establish the first \emph{time-efficient} algorithm for learning bosonic Gaussian unitaries under this metric, providing rigorous, end-to-end guarantees on the number of channel uses and the computational complexity required to learn the target Gaussian unitary to a prescribed accuracy.
Recall that arbitrary Gaussian unitaries can be expressed as the product of a displacement operator $D_{\mathbf r}$ and a symplectic Gaussian unitary $U_S$.
Our procedure simultaneously estimates the symplectic component $S \in \mathrm{Sp}_{2m}(\R)$ and the displacement vector $\mathbf r \in \R^{2m}$ using only Gaussian resources: coherent and two-mode–squeezed input probes, heterodyne detection, and a deterministic ``symplectic regularization'' step that projects an arbitrary estimate to a valid symplectic matrix. In addition, for displacement learning we also present, in~\cref{sec:algo-without-entanglement}, an alternative algorithm that employs only single-mode squeezed states, and in~\cref{sec:no-activesq} we describe how to avoid active squeezing operations. Taken together, these features suggest that our algorithms are experimentally efficient and feasible for current photonic platforms.

\medskip
We state our main result informally below and give the full theorem in~\cref{thm:end-to-end-diamond}.

\begin{theorem}[(Informal, see~\cref{thm:end-to-end-diamond} for details)]\label{thm:informal-main}
Let $m\in\mathds{N}$, $z\ge 1$, $\bar n>0$, $\bar n_{\mathrm{in}}>0$, 
$\varepsilon\in(0,1)$, and $\delta\in(0,1)$ be known parameters. 
There exists a quantum algorithm that
\begin{itemize}
    \item \textbf{\emph{Given:}} black-box access to an unknown $m$-mode Gaussian unitary $G_{\mathbf r,S} = D_{\mathbf r} U_S$, where $D_{\mathbf r}$ is the displacement operator and $U_S$ is a symplectic Gaussian unitary specified by a symplectic matrix $S$ with operator norm 
    $\|S\|_\infty \le z$;
    
    \item \textbf{\emph{Using:}} a runtime (and number of queries to $G_{\mathbf r,S}$) of
    \begin{equation}
        \mathsf{poly}(m, z, \bar n, \bar n_{\mathrm{in}}, 1/\varepsilon, \log(1/\delta)),
    \end{equation}
    and only input states with mean photon number of at most $\bar n_{\mathrm{in}}$;
    
    \item \textbf{\emph{Outputs:}} estimators $\tilde{\mathbf r} \in \R^{2m}$ and $\tilde S \in \mathrm{Sp}_{2m}(\R)$ such that the corresponding Gaussian unitary channel $\tilde{\mathcal G} \coloneqq \mathcal D_{\tilde{\mathbf r}} \circ \mathcal U_{\tilde S}$ approximates the true channel $\mathcal G \coloneqq \mathcal D_{\mathbf r} \circ \mathcal U_S$ and satisfies
    \begin{equation}
        \Pr\left[
            \frac{1}{2}\norm{\tilde{\mathcal{G}} - \mathcal{G}}_{\diamond,\bar n} 
            \le \varepsilon 
        \right] \ge 1-\delta,
    \end{equation}
    where the diamond norm is taken with respect to the mean photon number constraint $\bar n$.
\end{itemize}
\end{theorem}

The runtime of the algorithm is dominated by the number of queries, so the time complexity is asymptotically equivalent to the query complexity. This result provides the first end-to-end guarantee for learning general multi-mode Gaussian unitaries with finite-energy resources and shows that the task can be accomplished with resources that grow only polynomially in the natural parameters of the problem (e.g.~the number of modes $m$). Our algorithm also has the attractive feature that, if arbitrary input energy is available ($\bar n_{\mathrm{in}} \to \infty$), learning can be achieved to arbitrary precision with only $2m+2$ queries (see~\cref{remark:scaling-probepower-nu}).

\paragraph{Technical contributions and significance.}
To the best of our knowledge, the time-efficient symplectic regularization step is a novel technical contribution of our work. In contrast, the corresponding task in the fermionic setting, regularization to elements of the orthogonal group, is a convex optimization problem that can be solved very efficiently via the polar decomposition~\cite{oszmaniec2022fermion}.
Symplectic rounding, however, is inherently non-convex, and we overcome this computational obstacle while incurring only a multiplicative loss of $\mathcal{O}(\squeezingparameter^2)$ in the distance from a symplectic matrix.
This results in only a mild increase in the overall query complexity.

Even when disregarding computational efficiency, our work provides the first formal guarantees on the query complexity of learning bosonic Gaussian unitaries up to small energy-constrained diamond distance.
We believe that our results and techniques open the door to follow-up research, such as learning doped bosonic unitaries~\cite{mele2024learning} and other extensions of Gaussian process tomography.

\subsection{Technical overview}
We outline the main ideas behind our algorithm and its analysis, following the overall structure of the paper. In this overview we assume familiarity with the basics of continuous-variable quantum information, for which we provide a primer in~\cref{sec:pre-cv}. The precise problem formulation together with the parameter conventions and notation can be found in~\cref{sec:problem-define}.

\medskip
We specify an arbitrary Gaussian unitary $G_{\mathbf r, S} = D_{\mathbf r}U_S$  by a displacement vector $\mathbf r \in \R^{2m}$ and a symplectic matrix $S \in \mathrm{Sp}_{2m}(\R)$. The overall procedure can be summarized in four stages: 
\begin{enumerate}
    \item \emph{Symplectic estimation.}  
    Using coherent probes and heterodyne detection, we obtain an initial estimate of the symplectic action $S$ with operator-norm error at most~$\varepsilon_S$.  

    \item \emph{Symplectic regularization.}  
    A matrix square-root correction enforces exact symplecticity, producing $\tilde S \in \mathrm{Sp}_{2m}(\R)$ while preserving the accuracy guarantee $\|\tilde S-S\|_\infty \le \mathcal{O}(z^2\varepsilon_S)$, where $z \ge 1$ is a known squeezing bound with $\|S\|_\infty \le z$.  

    \item \emph{Displacement estimation.}  
    We employ squeezed Gaussian probes that cancel the symplectic action to first order (in particular $U_{\tilde{S}^{-1}}$).
    After this, we are able estimate the displacement vector and obtain $\tilde{\mathbf r}$ with Euclidean error $\varepsilon_r$.  

    \item \emph{End-to-end guarantee.}  
    Combining the symplectic and displacement estimates, the reconstructed channel $\tilde{\mathcal G} = \mathcal D_{\tilde{\mathbf r}} \circ \mathcal U_{\tilde S}$ satisfies $\tfrac12 \|\tilde{\mathcal G}-\mathcal G\|_{\diamond,\bar n} \le \varepsilon$ with probability at least $1-\delta$, provided $\varepsilon_S$ and $\varepsilon_r$ are chosen according to the bounds in~\cref{thm:end-to-end-diamond}.
\end{enumerate}

\paragraph{Symplectic learning via coherent probing~(\cref{sec:vac-share,sec:sym-probe}).}
To access the symplectic action of a Gaussian unitary $G_{\mathbf r, S}$, we probe the channel with coherent states and analyze the heterodyne detection outcomes at the output. For an input coherent state $\ket{\mathbf m}$ with mean vector $\mathbf m \in \R^{2m}$, the heterodyne outcomes are distributed as a Gaussian with mean $\mathbf r + S\mathbf m$ and covariance matrix $(SS^\top+\id)/2$. Thus, appropriately chosen probe vectors $\mathbf m$ allow one to recover the columns of $S$ via centered finite differences. 

In the \emph{vacuum-shared scheme}, one compares the heterodyne averages $\bar Y_i$ obtained from $\ket{\eta e_i}$ against the vacuum response $\bar Y_0$, where $e_i \in \R^{2m}$ is the $i$-th standard basis vector and $\eta>0$ is the probe amplitude. The resulting estimator is defined as $\hat S_i = (\bar Y_i - \bar Y_0)/\eta$, which provides an unbiased estimate of the $i$-th column $S_i$ of $S$. However, since all column estimators share the same random vacuum baseline $\bar Y_0$, their fluctuations are statistically dependent. This correlation complicates the operator-norm analysis, and one is forced to control each column separately and then take a union bound across all $2m$ columns.

In the \emph{symmetric-probe scheme}, by contrast, each column $S_i$ is estimated from two independent batches of heterodyne samples obtained from $\ket{+\eta e_i}$ and $\ket{-\eta e_i}$. The estimator $\hat S_i = (\bar Y_i^{(+)} - \bar Y_i^{(-)})/(2\eta)$ eliminates the displacement $\mathbf r$ and depends only on probe-specific noise. As a consequence, the estimators for different columns are mutually independent. This structural independence allows one to model the entire error matrix $\hat S - S$ as a standard Gaussian matrix after rescaling, thereby enabling the application of sharp operator-norm concentration bounds from random matrix theory. The resulting analysis gives tighter guarantees on $\|\hat S - S\|_\infty$ compared to the vacuum-shared scheme.

Both schemes provide high-probability guarantees on the operator-norm error $\|\hat S - S\|_\infty\le\varepsilon_S$, as shown in~\cref{thm:learnS-again,lem:pm-design}. The query complexity scales as $(2m+1)N_S$ in the vacuum-shared case and $4mN_S$ in the symmetric-probe case, where $N_S$ denotes the number of heterodyne shots per probe. In the asymptotic regime where the photon-number constraint on the input $\bar n_\text{in}$ is taken to infinity (equivalently $\eta \to \infty$), a single shot per probe suffices, yielding the asymptotic limits of $2m+1$ and $4m$ queries, respectively.

\paragraph{Symplectic regularization~(\cref{sec:close-sym}).}
The empirical estimator $\hat S$ is not guaranteed to be symplectic, and no efficient algorithm is currently known for projecting general matrices onto the symplectic group. To resolve this, we introduce a rounding procedure based on symplectic polar decompositions. Define $T \coloneqq -\Omega \hat S^\top \Omega \hat S$, which equals the identity when $\hat S$ is exactly symplectic. When $\hat S$ is close to $S$, the perturbation bound $\|T-\id\|_\infty \le (2\norm{S}_\infty+1)\|\hat S-S\|_\infty$ ensures that $T$ lies in a neighborhood where the principal square root $Q \coloneqq \sqrt{T}$ exists, is unique, and is Lipschitz continuous in operator norm. We then set $\tilde S \coloneqq \hat S Q^{-1}$, which is symplectic (i.e.~it satisfies $\tilde S^\top \Omega \tilde S = \Omega$) by construction. Matrix-analytic arguments further imply
\begin{equation}
    \norm{\tilde S - S}_\infty \le 9 \norm{S}_\infty^2 \, \norm{\hat S - S}_\infty,    
\end{equation}
showing that the regularization step incurs only a quadratic blow-up in $\|S\|_\infty$ when bounding the error. This yields an efficient rounding procedure that enforces exact symplecticity while preserving accuracy.

\paragraph{Displacement learning with two-mode squeezed states~(\cref{sec:dis-entangle}).}
We next consider the setting where auxiliary entanglement is available via two-mode squeezed vacuum (TMSV) probes of squeezing strength $\nu\ge1$. The protocol prepares $\ket{\nu}^{\otimes m}=U_{S_\nu}\ket{0}^{\otimes 2m}$, applies the preprocessing $U_{\tilde S^{-1}}$, the unknown Gaussian unitary $G_{\mathbf r,S}$, and finally the inverse squeezing $U_{S_\nu}^\dagger$. The first $2m$ output modes then have mean vector $\sqrt{\nu}\,\mathbf r$ and covariance $\Sigma^{(1)}=(A+\id)/2$, where
\begin{equation}
A=\id+\nu(\Delta+\Delta^\top)+\nu(2\nu+1)\Delta\Delta^\top,
\qquad \Delta=\tilde S^{-1}S-\id.
\end{equation}
Heterodyne detection on these modes yields Gaussian samples 
$Y^{(1)}\sim\mathcal N(\sqrt{\nu}\mathbf r,\Sigma^{(1)})$, from which an unbiased estimator $\tilde{\mathbf r}=\hat\mu^{(1)}/\sqrt\nu$ is obtained. The accuracy of this estimator is governed by the covariance $\Sigma^{(1)}$, and concentration bounds show that $\|\tilde{\mathbf r}-\mathbf r\|_2\le\varepsilon_r$ with high probability, provided the number of repetitions $N_r$ scales as given in~\cref{thm:learn-r-nu-only}. In particular, larger squeezing $\nu$ suppresses the estimator variance, but only if the mismatch satisfies $\norm{\Delta}_\infty= o(1/\nu)$, reflecting the balance between probe energy and symplectic accuracy. Thus, this entangled scheme provides a high-probability guarantee for estimating $\mathbf r$ under experimentally reasonable assumptions on $\tilde S$.

\paragraph{Experimental feasibility~(\cref{sec:no-activesq}).}  
Although the protocol seems to demand the online application of active Gaussian unitaries such as $U_{S_\nu}^\dagger$ and $U_{\tilde S^{-1}}$ to squeezed inputs, both steps can in fact be avoided. The apparent online use of $U_{S_\nu}^\dagger$ before heterodyne detection can be simulated entirely passively: one injects a fixed squeezed ancilla, interferes it with the signal via beam splitters, performs complementary homodyne measurements, and finally rescales the classical outcomes as $\mathbf q\mapsto S_\nu \mathbf q$ (see~\cref{prop:passive-het}). Likewise, the input $U_{\tilde S^{-1}}\ket{\nu}^{\otimes m}$ can be rewritten via the Bloch--Messiah decomposition as an offline-prepared multimode squeezed state $\ket{Z}$ followed by a passive interferometer $U_{O_1}$. 

Hence, the entire block $U_{S_\nu}^\dagger G_{\mathbf r,S} U_{\tilde S^{-1}} \ket{\nu}^{\otimes m}$ admits an implementation using only passive Gaussian optics, pre-prepared squeezed vacua, and homodyne detection. This eliminates the need for online active squeezing of the unknown signal, while reproducing exactly the heterodyne statistics assumed in our analysis. As a result, all theoretical guarantees on the displacement estimation---in particular the high-probability error bounds from~\cref{thm:learn-r-nu-only}---remain valid, but now in a form that is directly compatible with experimentally accessible resources.

\paragraph{Displacement learning with single-mode squeezed states~(\cref{sec:algo-without-entanglement}).}
The efficient algorithm described above for learning the displacement vector relies on probing the unknown Gaussian unitary with (two-mode squeezed vacuum) states entangled with an auxiliary system. One might be tempted to mistakenly attribute the algorithm’s efficiency to the use of such an ancilla---which remains unaffected by the unknown unitary---as is the case in other quantum learning tasks whose efficiency fundamentally depends on auxiliary systems~\cite{Chen_2022,Chen_2024,oh2024entanglementenabled,liu2025quantumlearningadvantagescalable}. However, we show that this is not the case. In fact, we now introduce an alternative protocol for estimating the displacement vector $\mathbf{r}$ that eliminates the need for ancilla-assisted probes and instead employs only single-mode squeezed states as inputs.  The idea is to separately access the momentum and position components of $\mathbf r$. To estimate the momentum part, one prepares $U_{\tilde S^{-1}}\ket{z_{\mathrm{in}}}^{\otimes m}$, a product of $m$ momentum-squeezed states of strength $z_{\mathrm{in}}\ge 1$, applies the unknown Gaussian unitary $G_{\mathbf r,S}$, and performs homodyne detection in the momentum quadrature. This yields samples $Y_p\sim \mathcal N(\mathbf r_p,V_{pp}/2)$, where $V_{pp}$ is the momentum–momentum block of $(\Delta+\id)V(\ket{z_{\mathrm{in}}})(\Delta+\id)^\top$ with $\Delta=S\tilde S^{-1}-\id$. Similarly, the position part is learned by preparing $U_{\tilde S^{-1}}\ket{z_{\mathrm{in}}^{-1}}^{\otimes m}$, applying $G_{\mathbf r,S}$, and performing homodyne detection in the position quadrature, which yields samples $Y_x\sim\mathcal N(\mathbf r_x,V_{xx}/2)$ with covariance block $V_{xx}$. The operator norm of the measurement covariance satisfies
\begin{equation}
    \max\left(\left\|\frac{V_{xx}}{2}\right\|_\infty, \left\|\frac{V_{pp}}{2}\right\|_\infty \right)\le \frac{1}{2}\left(\frac{(1+\|\Delta\|_\infty)^2}{z_{\mathrm{in}}}+z_{\mathrm{in}}\|\Delta\|_\infty^2\right).   
\end{equation} 
Hence, in the regime $\|\Delta\|_\infty=o(1/z_{\mathrm{in}})$ the statistical noise vanishes as $z_{\mathrm{in}}\to\infty$, enabling efficient estimation of the displacement vector $\textbf{r}$. 


\paragraph{End-to-end guarantee~(\cref{sec:end-to-end}).}
The full guarantee follows by combining the symplectic and displacement learning analyses with the diamond-norm error propagation.  
For the symplectic part, heterodyne-based estimation yields an additive error $\|\tilde S-S\|_\infty\le \varepsilon_S$, using either $(2m+1)N_S$ queries in the vacuum-shared scheme (\cref{sec:vac-share}) or $4mN_S$ in the symmetric-probe scheme (\cref{sec:sym-probe}).  
Through the additive-to-multiplicative conversion (\cref{lem:add-to-mult,lem:inv-perturb-symplectic}), this error induces a multiplicative mismatch $\Delta=\tilde S^{-1}S-\id$ of order $\|\Delta\|_\infty=\mathcal O(z\varepsilon_S)$.

For the displacement part, the measurement distributions arising from squeezed probes depend explicitly on $\Delta$.  
Substituting the above bound into the concentration inequalities (\cref{thm:learn-r-nu-only,query_displ2}) determines the required number of repetitions $N_r$ to guarantee $\|\tilde{\mathbf r}-\mathbf r\|_2\le \varepsilon_r$.  
The total cost then depends on the chosen protocol: with entangled two-mode squeezed probes, one needs $(\text{symplectic cost})+N_r$ queries, while with single-mode squeezed probes, both quadratures must be learned separately, leading to $(\text{symplectic cost})+2N_r$ queries.

These two errors combine at the channel level under the energy-constrained diamond distance.  
By~\cref{metalemma}, we obtain
\begin{equation}
    \frac12\|\tilde{\mathcal G}-\mathcal G\|_{\diamond,\bar n}
 \le 12\sqrt{9\sqrt{2m}(\bar n+1)}\,\sqrt{z\sqrt{2m}\,\varepsilon_S}
   + \sqrt{2}\sqrt{z^2\bar n+1}\,\varepsilon_r,    
\end{equation}
where the first term arises from controlling the effect of the symplectic mismatch and the second accounts for displacement errors. Thus, choosing 
\begin{equation}
    \varepsilon_S\le \frac{\varepsilon^2}{2592mz(\bar n+1)}, \qquad \varepsilon_r\le \frac{\varepsilon}{2\sqrt{2}\sqrt{z^2\bar n+1}},    
\end{equation}
ensures $\tfrac12\|\tilde{\mathcal G}-\mathcal G\|_{\diamond,\bar n}\le \varepsilon$ with success probability at least $1-\delta$.

Altogether, this yields four distinct end-to-end learning algorithms, depending on the choice of symplectic and displacement routines.  
Each achieves polynomial dependence on $m$, $z$, $\bar n$, $\bar n_{\mathrm{in}}$, $1/\varepsilon$, $\log(1/\delta)$, thereby providing the first rigorous, time- and sample-efficient guarantees for learning multiparameter bosonic Gaussian unitaries.  
In the high-energy limit $\bar n_{\mathrm{in}}\to\infty$, both $N_S$ and $N_r$ converge to~1.  
In particular, the vacuum-shared symplectic routine combined with two-mode squeezed displacement learning yields the fundamental query complexity $N_{\mathrm{tot}}=2m+2$.  
The other protocol combinations can be treated by the same analysis, leading to analogous query complexities.

\subsection{Related work}
The problem of learning the parameters of Gaussian channels has been a central focus of quantum estimation theory since the inception of the field by Caves~\cite{Caves1981}. The algorithm we propose is related to several ideas from established protocols, which we briefly review below.

\paragraph{Local and global estimation.} A vast literature has studied the optimal precision asymptotically attainable for the problem of local estimation, where the parameters are already approximately known; see the reviews~\cite{PRXQuantum.3.010202,RevModPhys.90.035006,Zhang_2021}. In particular, we highlight the recent work of Chang, Genoni, and Albarelli~\cite{chang2025multiparameterquantumestimationgaussian}, which provides a comprehensive account of the fundamental limits of Gaussian multiparameter estimation. D'Ariano et al.~\cite{D_Ariano_2001} observed that a two-mode squeezed vacuum probe can estimate a displacement with arbitrarily high precision, a problem further studied in~\cite{PhysRevA.87.012107,PhysRevA.97.012106}. Global parameter estimation of Gaussian processes in the Bayesian setting was developed by Morelli et al.~\cite{Morelli_2021}, though their results do not provide guarantees for worst-case scenarios. The problem of optimal estimation of group action by minimizing an expected cost function was studied in~\cite{Hayashi2016} in generality and with applications to the infinite dimensional case with an energy constraint. Finally, Rosati~\cite{Rosati2024learningtheory} established sample-complexity bounds for learning the output distribution of a CV circuit from the perspective of statistical learning theory.

\paragraph{Learning the symplectic part of Gaussian unitaries.} Several works have proposed reconstruction methods for CV channels. Closest to this paper, the recent work of Grove et al.~\cite{Grove2025} presents an algorithm for learning the symplectic matrix of a Gaussian unitary (but not the displacement), without a sample-complexity analysis. There are two main differences between their approach and ours: (i) we use heterodyne detection instead of homodyne detection, since at high energies the variance reduction from homodyne detection is offset by the fact that it requires twice as many measurements; and (ii) we output a matrix that is arbitrarily close to a true symplectic matrix, whereas no such guarantee is provided in~\cite{Grove2025}.

\paragraph{Process tomography with coherent probes.} Lobino et al.~\cite{Lobino2008} proposed and experimentally implemented a method for general single-mode process tomography in truncated Fock space, using coherent states as inputs and homodyne detection to reconstruct the channel action in the Glauber--Sudarshan representation. Rahimi-Keshari et al.~\cite{Rahimi2010} modified this method to avoid explicit use of the Glauber--Sudarshan representation and extended it to the multimode setting. Subsequently, Wang et al.~\cite{Wang2013} proposed a reconstruction method for Gaussian processes based on the Husimi $Q$ function, requiring $2m+1$ distinct coherent-state inputs. Ghalaii and Rezakhani~\cite{Ghalaii2017} developed a reconstruction procedure using normally ordered moments, while Kumar et al.~\cite{Kumar2020_2} considered photon-number detection measurements.

\paragraph{Connections to related learning tasks.} 
Similar learning algorithms also exist for conceptually related classes of unitary operators, such as Clifford and fermionic Gaussian unitaries~\cite{low2009learning,lai2022learning,leone2024learning,oszmaniec2022fermion}. Like bosonic Gaussian unitaries, these operators admit concise classical representations and possess algebraic properties that make them efficiently learnable. In contrast, the discrete-variable setting simplifies the associated learning algorithms. Moreover, fermionic Gaussian unitaries are associated with the orthogonal group, whereas bosonic Gaussian unitaries are associated with the symplectic group. While efficient algorithms exist to find the closest orthogonal matrix (in operator norm) to a given matrix, no such algorithm is known for the symplectic group. Overcoming this barrier is a key technical contribution of our work.

On the state side, the problem of learning Gaussian states has already been explored extensively both for bosons and fermions~\cite{mele2024learning,bittel2024optimalestimatestracedistance,fanizza2024efficienthamiltonianstructuretrace,bittel2025energyindependenttomographygaussianstates,ogorman2022fermionictomographylearning,aaronson2023efficienttomographynoninteractingfermion,Bittel_2025Optimalferm}, with recent notable advances also extending to non-Gaussian regimes~\cite{mele2024learning,Mele_2025ferm,iosue2025higher}.

\subsection{Discussion and open problems}
In this work, we show that bosonic Gaussian unitaries can be efficiently learned, with both query and time complexity scaling polynomially in the number of modes. This resolves an open question raised in recent works~\cite{bittel2024optimalestimatestracedistance,bittel2025energyindependenttomographygaussianstates}. The proposed learning algorithm is experimentally feasible and can be implemented on current photonic platforms. We provide a detailed complexity analysis with rigorous error guarantees, formulated in terms of the physically motivated energy-constrained diamond norm~\cite{VV-diamond,Shirokov2018,PLOB,EC-diamond}. A key feature is that, in the high-energy regime, the query complexity scales linearly with the number of modes while remaining constant in all other parameters (e.g., the squeezing of the unknown Gaussian unitary), thereby exhibiting a form of ``squeezing independence'' analogous to the ``energy independence'' identified in~\cite{bittel2025energyindependenttomographygaussianstates} for learning Gaussian states. Technically, our approach combines tools from quantum information theory with continuous-variable systems, concentration inequalities for Gaussian random matrices, and a novel method for truncating approximate symplectic matrices into valid ones while controlling the truncation error.

\medskip
This work opens a number of fundamental questions:
\begin{itemize}
    \item Our analysis concerns only unitary Gaussian channels. A natural question is whether general (non-unitary) bosonic Gaussian channels can also be efficiently learned. The main challenge lies in establishing perturbation bounds on the energy-constrained diamond distance between Gaussian channels, analogous to those available in the unitary case~\cite{EC-diamond}. Can the ``derivative approach'' of~\cite{bittel2024optimalestimatestracedistance} be adapted to derive such bounds?

    \item An important open problem is whether matching query-complexity upper and lower bounds can be established for learning arbitrary bosonic Gaussian unitaries, accounting for all problem parameters.
    In particular, our analysis shows that the efficient algorithm presented here achieves an upper bound scaling as $1/\varepsilon^{4}$ in the target accuracy $\varepsilon$, measured in the energy-constrained diamond norm, a dependence that follows from current estimates of this distance between Gaussian unitaries~\cite{EC-diamond}.
    Sharper bounds on this metric could potentially reduce the scaling to $1/\varepsilon^{2}$, as already achieved for fermionic Gaussian unitaries~\cite{oszmaniec2022fermion} and for learning fermionic and bosonic Gaussian states~\cite{bittel2024optimalestimatestracedistance}.
    
    \item What are the optimal query bounds for learning arbitrary (non-Gaussian) bosonic unitaries? The corresponding discrete-variable problem was recently addressed in~\cite{haah2023query}, while the continuous-variable analogue in the setting of state learning was solved in~\cite{mele2024learning}.
    
    \item What is the trade-off between the efficiency of learning bosonic unitaries and the degree of non-Gaussianity in the unknown unitary? One possible approach is to study tomography of \emph{$t$-doped} Gaussian unitaries~\cite{mele2024learning,Mele_2025}, i.e., unitaries that can be implemented using arbitrary Gaussian unitaries together with at most $t$ single-mode non-Gaussian unitaries. The corresponding state learning problem is efficient for $t = \mathcal{O}(1)$~\cite{mele2024learning}, but extending this to unitary learning presents new challenges. Nonetheless, the analogous problem for doped fermionic Gaussian unitaries has been solved~\cite{iyer2025mildly,austin2025efficiently}, suggesting that similar methods may be effective in the bosonic setting.
    
    \item Finally, the result in this work is proved in the so-called \emph{realizable} setting, where the input is assumed to be exactly a Gaussian unitary. A related open problem is to address the \emph{agnostic} setting: given a bosonic unitary that is $\varepsilon$-close to some Gaussian unitary (in an energy-constrained average-case metric), can we output a Gaussian unitary that is $(\varepsilon + \varepsilon')$-close to the input? Here, $\varepsilon'$ represents the additional approximation overhead beyond the best approximation achievable within the class of Gaussian unitaries, which arises because the true input need not lie in this class. This problem, known as \emph{agnostic tomography}, has been intensively studied in the discrete-variable setting~\cite{grewal2024agnostic,grewal2024improved,chen2024stabilizer,bakshi2024learning,wadhwa2024agnostic, Badescu2021, Fanizza2024, fanizza2025learningfinitelycorrelatedstates}.
\end{itemize}

\section{Preliminaries}

\paragraph{Notation.} 
We write $\mathds{R}$, $\mathds{N}$, and $\mathds{C}$ for the sets of real, natural, and complex numbers, respectively. 
For a matrix $M$, we denote by $\norm{M}_\infty$ its operator (spectral) norm, by $\norm{M}_2$ its Hilbert--Schmidt norm, and by $\norm{M}_1$ its trace norm. 
For an integer $n \ge 1$, we use the shorthand $[n] \coloneqq \{1,\ldots,n\}$. 
$\id$ denotes the identity operator. 
We write $\delta_{ij}$ for the Kronecker delta, which equals $1$ if $i=j$ and $0$ otherwise.

Given a mean vector $\mathbf m \in \R^k$ and a positive definite covariance matrix $V \in \R^{k \times k}$, we write $\mathcal{N}(\mathbf m, V)$ for the Gaussian probability distribution on $\R^k$ with mean $\mathbf m$ and covariance $V$. Its probability density function at $\mathbf x \in \R^k$ is
\begin{equation}
    \mathcal{N}(\mathbf m, V)(\mathbf x) \coloneqq 
    \frac{e^{-\tfrac{1}{2}(\mathbf x - \mathbf m)^{\top} V^{-1} (\mathbf x - \mathbf m)}}{(2\pi)^{k/2}\sqrt{\det V}}.
    \label{eq:gaussian-def}
\end{equation}
We denote by $\mathrm{Sp}_{2m}(\R)$ the real symplectic group, i.e., the set of all $2m \times 2m$ real matrices $S$ satisfying $S^{\top} \Omega S = \Omega$, where 
\begin{equation}\label{eq:omega}
    \Omega \coloneqq \bigoplus_{i=1}^m
        \begin{pmatrix}
            0 & 1\\
            -1 & 0
        \end{pmatrix}
\end{equation}
is the canonical symplectic form. We write $\mathrm{O}(n)$ for the orthogonal group of degree $n$, consisting of all real $n \times n$ matrices $Q$ such that $Q^{\top}Q = \id$.

\subsection{Continuous-variable systems}\label{sec:pre-cv}
In this section, we review the relevant preliminaries on continuous-variable (CV) systems; for further details, we refer to the book by Serafini~\cite{BUCCO}.

A continuous-variable system is a quantum system associated with the so-called $m$-mode Hilbert space $L^2(\R^m)$, consisting of all square-integrable complex functions over $\R^m$. Here $m\in\mathds{N}$ denotes the number of modes, which play the role of the \emph{system size}, analogous to the number of qudits in discrete-variable (DV) systems. Note that the number of modes $m$ can be regarded as the system size of a CV system because $L^2(\R^m) = \bigl(L^2(\R)\bigr)^{\otimes m}$.

A quantum state acting on this Hilbert space is called an $m$-mode quantum state, and a unitary operator acting on it is called an $m$-mode unitary. On this space, one can define the usual quadrature operators: for each $i\in[m]$, let $\hat{x}_i$ and $\hat{p}_i$ denote the position and momentum operators of the $i$-th mode. It is convenient to collect them into the \emph{quadrature operator vector}
\begin{equation}
    \hat{\mathbf{R}} \coloneqq (\hat{x}_1,\hat{p}_1,\ldots,\hat{x}_m,\hat{p}_m).
\end{equation}
The quadratures satisfy the canonical commutation relations $[\hat{\mathbf R}_j,\hat{\mathbf R}_\ell] = i\Omega_{j\ell}\id$ for all $j,\ell\in[2m]$, where $\id$ denotes the identity, and $\Omega \in \R^{2m\times 2m}$ is defined in~\cref{eq:omega}. Equivalently, these relations can be written compactly as $[\hat{\mathbf{R}},\hat{\mathbf{R}}^\top] = i\Omega\id$.

\paragraph{Gaussian unitaries.}
Gaussian unitaries are $n$-mode unitaries generated by Hamiltonians quadratic in the quadrature operators. The set of $n$-mode Gaussian unitaries is in one-to-one correspondence with pairs $(\mathbf{r},S)$, where $\mathbf{r}\in\R^{2m}$ and $S\in\mathrm{Sp}_{2m}(\R)$.
For a Gaussian unitary $G_{\mathbf r,S}$ specified by $(\mathbf{r},S)$, the vector $\mathbf{r}$ is called the \emph{displacement vector} and $S$ the \emph{symplectic matrix}. More precisely, any Gaussian unitary admits the decomposition
\begin{equation}
    G_{\mathbf{r},S} \coloneqq D_{\mathbf{r}} U_S,
\end{equation}
where $D_{\mathbf{r}} \coloneqq e^{-i\mathbf{r}^\top \Omega \hat{\mathbf{R}}}$ is the displacement operator, and $U_S$ is the symplectic Gaussian unitary associated with $S$. Their action on quadratures is
\begin{align}
    D_{\mathbf{r}}^\dagger \hat{\mathbf{R}} D_{\mathbf{r}} &= \hat{\mathbf{R}} + \mathbf{r}\id, \\
    U_S \hat{\mathbf{R}} U_S^\dagger &= S \hat{\mathbf{R}}.
\end{align}

\paragraph{Gaussian states.}
Gaussian states are $m$-mode states that can be expressed as tensor products of Gibbs states of quadratic Hamiltonians in the quadrature operator vector. Equivalently, they are precisely the states preserved under Gaussian unitaries. Any Gaussian state is uniquely characterized by its \emph{first moment} and \emph{covariance matrix}, defined for a state $\rho$ as
\begin{align}
    \mathbf{m}(\rho) & = \Tr[\hat{\mathbf{R}} \rho] \in \R^{2m}, \\
    V(\rho) &= \Tr\left[\left\{ \hat{\mathbf{R}} - \mathbf{m}(\rho)\id, \; \bigl(\hat{\mathbf{R}} - \mathbf{m}(\rho)\id\bigr)^\top \right\}\rho \right] \in \R^{2m\times 2m},
\end{align}
where $\{A,B\} = AB + BA$ denotes the anticommutator. Displacements act on moments as
\begin{align}
    \mathbf{m} \bigl(D_{\mathbf{r}}\rho D_{\mathbf{r}}^\dagger\bigr) &= \mathbf{m}(\rho)+\mathbf{r},\\
    V\bigl(D_{\mathbf{r}}\rho D_{\mathbf{r}}^\dagger\bigr) &= V(\rho),
\end{align}
while symplectic Gaussian unitaries act as
\begin{align}
    \mathbf{m}\bigl(U_S\rho U_S^\dagger\bigr) &= S\mathbf{m}(\rho), \\
    V\bigl(U_S\rho U_S^\dagger\bigr) &= S V(\rho) S^\top.
\end{align}

\emph{Coherent states} are examples of pure Gaussian states. They are denoted by $\ket{\mathbf{m}}$, with first moment $\mathbf{m}(\ket{\mathbf{m}}) = \mathbf{m}$ and covariance matrix $V(\ket{\mathbf{m}}) = \id$. The coherent state with $\mathbf{m}=0$ is called the \emph{vacuum state}.  

The \emph{single-mode squeezed state} $\ket{z_\mathrm{in}}$, with squeezing parameter $z_\mathrm{in}>0$, is a pure Gaussian state with zero first moment and covariance matrix
\begin{equation}\label{eq_singlemode_squeezed}
    V(\ket{z_\mathrm{in}}) = 
    \begin{pmatrix}
        z_\mathrm{in} & 0 \\[4pt]
        0 & z_\mathrm{in}^{-1}
    \end{pmatrix}.
\end{equation}
More generally, the $m$-mode product state $\ket{z_\mathrm{in}}^{\otimes m}$, squeezed in the momentum quadrature by factor $z$, has covariance matrix
\begin{equation}
    V(\ket{z_\mathrm{in}}^{\otimes m}) = \bigoplus_{i=1}^m 
    \begin{pmatrix}
        z_\mathrm{in} & 0 \\[4pt]
        0 & z_\mathrm{in}^{-1}
    \end{pmatrix}.
\end{equation}

The \emph{two-mode squeezed vacuum} (TMSV) state $\ket{\nu}$, where $\nu\geq 1$, is a two-mode pure Gaussian state with zero first moment, and covariance matrix
\begin{equation}\label{eq_tmsv}
    V(\ket{\nu}) =
    \begin{pmatrix}
        (2\nu-1)\id & 2\sqrt{\nu(\nu-1)}\,\sigma_z \\[2pt]
        2\sqrt{\nu(\nu-1)}\,\sigma_z & (2\nu-1)\id
    \end{pmatrix},
    \qquad
    \sigma_z \coloneqq \begin{pmatrix} 1 & 0 \\ 0 & -1 \end{pmatrix}\,,
\end{equation}
written in $2\times 2$ block form.

\paragraph{Energy operator and mean energy.}
The energy operator is defined as
\begin{equation}\label{eq:energy-operator}
    \hat{E} \coloneqq \frac{\hat{\mathbf{R}}^\top \hat{\mathbf{R}}}{2} = \sum_{i=1}^m \frac{\hat{x}_i^2 + \hat{p}_i^2}{2}.
\end{equation}
The mean energy of a state $\rho$ is $E(\rho) \coloneqq \Tr[\hat{E}\rho]$, which can be written in terms of moments and covariance as
\begin{equation}
    E(\rho) = \frac{\Tr\left[V(\rho)\right]}{4} + \frac{\norm{\mathbf{m}(\rho)}_2^2}{2}.
\end{equation}

\paragraph{Photon number operator.}
For each mode $j\in[m]$, let $a_j$ and $a_j^\dagger$ denote the annihilation and creation operators, related to the quadratures by
\begin{align}
    a_j & \coloneqq \frac{\hat{x}_j + i\hat{p}_j}{\sqrt{2}}, \\
    a_j^\dagger & \coloneqq \frac{\hat{x}_j - i\hat{p}_j}{\sqrt{2}},
\end{align}
which satisfy the canonical commutation relations $[a_i,a_j^\dagger] = \delta_{ij}$. The \emph{total photon number operator} (or \emph{total particle number operator}) is defined as
\begin{equation}
    \hat{N} \coloneqq \sum_{i=1}^m a_i^\dagger a_i.
\end{equation}

\noindent The particle number operator is related to the energy operator $\hat{E}$ by
\begin{equation}
    \hat{E} = \hat{N} + \frac{m\id}{2},
\end{equation}
so that the mean energy of a state $\rho$ decomposes as
\begin{equation}
    E(\rho) = \langle \hat{N} \rangle_\rho + \frac{m}{2},
\end{equation}
or, equivalently,
\begin{equation}\label{eq_mean_photon}
    \langle \hat{N} \rangle_\rho = \frac{\Tr[V(\rho)-\id]}{4} + \frac{\norm{\mathbf{m}(\rho)}_2^2}{2}.
\end{equation}

\paragraph{Homodyne measurement.} The homodyne measurement is one of the most standard measurement techniques in quantum optics, frequently used to access the quadratures of bosonic modes. Consider an $n$-mode system described by the quadrature operator vector $\hat{\mathbf{R}} = (\hat x_1, \hat p_1, \ldots, \hat x_n, \hat p_n)^\top.$ A homodyne measurement consists of selecting, for each mode $j \in [n]$, whether to measure its position operator $\hat x_j$ or its momentum operator $\hat p_j$.

Suppose a Gaussian state $\rho$ is characterized by first-moment vector $\mathbf{m} \in \R^{2n}$ and covariance matrix $V \in \R^{2n \times 2n}$. 
When $\rho$ is measured by homodyne detection, the outcomes are distributed according to a multivariate Gaussian law whose mean and covariance are obtained by restricting $m$ and $V/2$ to the measured quadratures. 
Two important special cases are as follows:
\begin{itemize}
    \item If all position quadratures $(\hat x_1,\ldots,\hat x_n)$ are measured, the outcomes $\mathbf{x} \in \R^n$ are distributed as $\mathbf{x} \sim \mathcal N(\mathbf{m}_x,\, V_{xx}/2)$, where $\mathbf{m}_x \coloneqq (m_{2j-1}(\rho))_{j \in [n]}$ is the vector of odd entries of $m$, and $V_{xx} \coloneqq ( V_{2i-1,\,2j-1} )_{i,j \in [n]}$
    is the $n \times n$ submatrix of $V$ corresponding to position-position covariances.

    \item If all momentum quadratures $(\hat p_1,\ldots,\hat p_n)$ are measured, the outcomes $\mathbf{p} \in \R^n$ are distributed as $\mathbf{p} \sim \mathcal N(\mathbf{m}_p,\, V_{pp}/2)$, where $\mathbf{m}_p \coloneqq (\mathbf{m}_{2j}(\rho))_{j \in [n]}$ is the vector of even entries of $\mathbf{m}$, and $V_{pp} \coloneqq ( V_{2i,\,2j} )_{i,j \in [n]}$ is the $n \times n$ submatrix of $V$ corresponding to momentum-momentum covariances.
\end{itemize}

From an experimental perspective, homodyne measurements are implemented by interfering the signal mode with a strong local oscillator at a beam splitter, followed by measuring the output intensities. 
The phase of the local oscillator determines the quadrature being measured: a phase of $0$ corresponds to $\hat x$, a phase of $\pi/2$ to $\hat p$, and intermediate values to rotated quadratures.

\paragraph{Heterodyne measurement.}\label{def:hetero}
The heterodyne measurement provides simultaneous (but noisy) information about both the position and momentum quadratures of each bosonic mode. 
For a Gaussian state $\rho$ with first-moment vector $\mathbf m \in \R^{2m}$ and covariance matrix $V \in \R^{2m\times 2m}$, the heterodyne outcomes are distributed according to a classical Gaussian law,
\begin{equation}
    \mathsf{Heterodyne}(\rho) \sim \mathcal{N}\left(\mathbf m,\, \frac{V +\id}{2} \right),
    \label{eq:heterodyne-gaussian}
\end{equation}
where $\mathcal N(\cdot,\cdot)$ is the Gaussian probability distribution defined in~\cref{eq:gaussian-def}.

In practice, heterodyne detection is realized by interfering the $m$-mode state $\rho$ with an auxiliary $m$-mode vacuum state on a balanced beam splitter, and subsequently performing homodyne measurements on both output ports: the position quadratures are measured on one arm, and the momentum quadratures on the other. 

\subsection{Linear algebra and matrix analysis}
We review several results from linear algebra and matrix analysis that will be useful in our proofs.

\medskip
First, we discuss the Euler decomposition of a symplectic matrix.
\begin{lemma}[(Euler decomposition)]\label{lem:euler-dec}
Any symplectic matrix $S \in \mathrm{Sp}_{2m}(\R)$ admits the Euler decomposition
\begin{equation}\label{Euler_dec2}
    S = O_1 Z_S O_2,
\end{equation}
where $O_1,O_2\in \mathrm{O}(2m) \cap \mathrm{Sp}_{2m}(\R)$ are symplectic orthogonal matrices, and
\begin{equation}\label{sq_matttt}
    Z_S \coloneqq \bigoplus_{j=1}^m 
    \begin{pmatrix} z_j & 0 \\[2pt] 0 & z_j^{-1} \end{pmatrix}, 
    \qquad z_j \ge 1,
\end{equation}
is the \emph{squeezing matrix}. The quantity $\norm{S}_\infty = \max_{j \in [m]} z_j$ is referred to as the \emph{squeezing parameter} of $S$. 
\end{lemma}

At the unitary level, this yields the decomposition of the corresponding symplectic Gaussian unitary $U_S = U_{O_1} U_{Z_S} U_{O_2}$, where $U_{Z_S}$, called the \emph{squeezing unitary}, is the Gaussian unitary associated with $Z_S$, and $U_{O_1}$ and $U_{O_2}$ are the Gaussian unitaries associated with the symplectic orthogonal matrices $O_1$ and $O_2$. Such unitaries are called \emph{passive}, as they preserve the energy operator defined in~\cref{eq:energy-operator}.

\medskip
Neumann series generalize the notion of geometric series to the setting of operators.
\begin{fact}[(Neumann series)]\label{fact:neumann}
    Let $A$ be a bounded linear operator on a normed space with $\norm{A}_{\infty} < 1$. 
    Then the series
    \begin{equation}
        \sum_{k=0}^\infty A^k
    \end{equation}
    converges in operator norm, to $(\id - A)^{-1}$.
    In particular, $(\id-A)$ is invertible with inverse given by the Neumann series above.
\end{fact}

In our symplectic regularization procedure, it will be necessary to work with 
matrix square roots. For this purpose, we recall the following standard result
on the existence and uniqueness of the principal square root.

\begin{fact}[(Existence and uniqueness of the principal square root, see e.g.~{\cite[Theorem~1.29]{higham2008functions}})]\label{thm:principal-root}
Let $A \in \C^{n\times n}$ have no eigenvalues on $\R^-$. Then there exists a unique square root $X$ of $A$ whose eigenvalues all lie in the open right half-plane. This $X$ is called the principal square root of $A$, and we write $X = \sqrt{A}= A^{1/2}$. If $A$ is real, then $\sqrt{A}$ is also real.
\end{fact}

In practice, when $A \in \R^{n\times n}$ has eigenvalues confined to the open right half-plane, 
the principal root can be written in terms of the matrix logarithm as $\sqrt{A} = \exp\left(\tfrac{1}{2}\log(A)\right)$, which immediately ensures the defining property $\sqrt{A}\cdot\sqrt{A} = A$. More generally, the principal $p$-th root admits the following useful integral representation:

\begin{lemma}[(Matrix $p$-th root representation and perturbation, see e.g.~{\cite[Theorem~2.1]{cardoso2012computation}})]\label{thm:root-integral-rep}
Let $A \in \R^{n \times n}$ have no eigenvalues on the closed negative real axis and $p > 1$.  
Then for any $r \ge (2\norm{A}_\infty)^{1/p}$,
\begin{equation}
A^{1/p}
= \frac{p \sin(\pi/p)}{\pi} \cdot
A \left(
\int_{0}^{r} (t^{p} \id + A)^{-1}  dt
+ \int_{r}^{\infty} (t^{p} \id + A)^{-1}  dt
\right),
\end{equation}
where
\begin{equation}
    \NORM{\int_{r}^{\infty} (t^{p} \id + A)^{-1} dt}_\infty \le \frac{2 r^{1-p}}{p-1}.    
\end{equation}
Moreover, if $P \in \R^{n \times n}$ is such that $A+P$ has no eigenvalues on $\R^-$, then
\begin{equation}
(A+P)^{1/p} - A^{1/p}
= \frac{p \sin(\pi/p)}{\pi}
\int_{0}^{\infty}
t^{p} \left(t^{p}\id + A + P\right)^{-1} P \left(t^{p}\id + A\right)^{-1}  dt.
\end{equation}
\end{lemma}

\subsection{Concentration inequalities}
Next, we recall a standard concentration bound for Gaussian random matrices, which will be a key tool for analyzing the operator norm.

\begin{lemma}[(Tail bound for Gaussian operator norm; rescaled version of~{\cite[{Theorem~II.13}]{DavidsonSzarek2001}})]
\label{lem:gauss-opnorm-infty}
Let $A\in\R^{m\times n}$ have i.i.d.\ $\mathcal{N}(0,1)$ entries. Then, for all $t\ge 0$,
\begin{equation}
    \Pr\left[\norm{A}_{\infty} \ge \sqrt{m}+\sqrt{n}+t\right] \le e^{-t^2/2}.
\end{equation}
\end{lemma}

\begin{proof}
Assume $m\le n$ (otherwise swap $m,n$) and set $\Gamma:=A/\sqrt{n}$ and $\beta:=m/n$.
Davidson--Szarek~\cite{DavidsonSzarek2001} state that, for all $\tau>0$,
\begin{equation}
    \Pr\left[s_1(\Gamma)\ge 1 + \sqrt{\beta} + \tau\right] \le e^{-n\tau^2/2}.
\end{equation}
Since $s_1(\Gamma) = s_1(A)/\sqrt{n} = \norm{A}_{\infty}/\sqrt{n}$ and $\sqrt{n}(1+\sqrt{\beta})=\sqrt{n}+\sqrt{m}$, substituting $\tau=t/\sqrt{n}$ yields the claimed inequality.
\end{proof}

The next concentration bound will also be useful.
\begin{lemma}[(Estimation of first moments, see e.g.~{\cite[Eq.~(1.1)]{lugosi2017subgaussianestimatorsmeanrandom}})]
\label{lem:estrelative}
Let $\{\hat x_i\}_{i=1}^{N}$ be $N$ i.i.d.~samples from an $n$-dimensional Gaussian distribution 
$\mathcal N(\mu,\Sigma)$. Define the empirical mean as $\hat \mu  \coloneqq \frac{1}{N}\sum_{i=1}^{N} \hat x_i$. 
Then, for every $\delta \in (0,1)$,
\begin{align}
\Pr \left[ \norm{\hat \mu - \mu}_2 \,  \le \frac{\chi_{n,\delta}}{\sqrt{N}} \sqrt{\norm{\Sigma}_\infty} \right] & \ge 1-\delta, 
\end{align}
where $\chi_{n,\delta} \coloneqq \sqrt{n} + \sqrt{2\log(1/\delta)}$.
\end{lemma}

\subsection{Energy-constrained diamond distance} \label{sec:energy-constrained-diamond-distance}
The setting of continuous-variable quantum systems necessitates a modification to the standard discrete-variable distance metrics. This section serves as an exposition to the so-called \emph{energy-constrained diamond distance}.

In quantum learning theory, quantifying the error in the estimation of quantum objects (e.g., states or processes) requires distance measures with a clear physical interpretation.
For quantum state tomography, the most widely used notion of distance is the trace distance~\cite{MARK, Sumeet_book}. Given two quantum states $\rho_1$ and $\rho_2$, the trace distance is defined as 
\begin{equation}
    \frac{1}{2}\norm{\rho_1 - \rho_2}_{1}.
\end{equation}
This measure has a strong operational meaning due to the Holevo--Helstrom theorem: if one is given a single copy of an unknown state, promised to be either $\rho_1$ or $\rho_2$ with equal prior probability $1/2$, the optimal success probability of correctly identifying the state is 
\begin{equation}
    \frac{1}{2}\left(1 + \frac{1}{2}\norm{\rho_1 - \rho_2}_{1}\right).
\end{equation}
Moreover, if two states are close in trace distance, then the expectation values of any bounded observable on them must also be close.

For tomography of discrete-variable (DV) quantum channels, the most common distance measure is the diamond distance~\cite{MARK,Sumeet_book}. For two channels $\Phi_1$ and $\Phi_2$, it is defined as
\begin{equation}
    \frac{1}{2}\norm{\Phi_1 - \Phi_2}_{\diamond} \coloneqq 
    \sup_{\rho_{AA'}} \frac{1}{2}\norm{\Id_{A}\otimes \Phi_1\left(\rho_{AA'}\right) - \Id_{A}\otimes \Phi_2\left(\rho_{AA'}\right)}_{1},
\end{equation}
where the supremum is over all input states $\rho_{AA'}$ of the channel system $A'$ together with an arbitrary auxiliary  system $A$. This distance also has a direct operational interpretation via the Holevo--Helstrom theorem: if one has access to a single use of an unknown channel, promised to be either $\Phi_1$ or $\Phi_2$ with equal prior probability $1/2$, then the maximum probability of correctly identifying the channel is 
\begin{equation}
    \frac{1}{2}\left(1 + \frac{1}{2}\norm{\Phi_1 - \Phi_2}_{\diamond}\right).
\end{equation}
Operationally, the maximum is taken over all protocols consisting of preparing an input state (possibly entangled with an auxiliary system), sending it through the unknown channel, and performing a measurement. Furthermore, if two channels are close in diamond distance, then the expectation values of any bounded observable on their outputs are also close, regardless of the chosen input state. Hence, the diamond distance provides a physically meaningful notion of distance for DV channels.

The situation is fundamentally different for continuous-variable channels. In this case, the diamond distance is \emph{not} physically meaningful. The issue is that its definition involves optimization over input states with unbounded energy, which are not physically realizable. This feature is absent in the DV case. Even more strikingly, there exist CV unitaries that are intuitively very similar, yet maximally far apart in diamond distance. For example, the diamond distance between any two distinct beam splitters is maximal:
\begin{equation}
\frac{1}{2}\norm{\mathcal{U}_{\lambda_1} - \mathcal{U}_{\lambda_2}}_\diamond = 1 
\quad \text{for all } \lambda_1,\lambda_2\in[0,1], \, \lambda_1 \neq \lambda_2,
\end{equation}
where $\mathcal{U}_\lambda$ denotes the beam splitter unitary with transmissivity $\lambda$~\cite{BUCCO}. This follows from the fact that one can input coherent states of arbitrarily high energy, as shown in~\cite{VV-diamond}. Thus, even two beam splitters with almost identical transmissivity are maximally far in diamond distance. This counterintuitive feature stems precisely from the inclusion of arbitrarily high-energy input states in the optimization.

To address this problem, Winter~\cite{VV-diamond} and Shirokov~\cite{Shirokov2018} introduced the \emph{energy-constrained diamond distance}, refining the original (slightly different) version proposed in Ref.~\cite{PLOB}. Given a parameter $\bar{n} > 0$ and two CV quantum channels $\Phi_1$ and $\Phi_2$, it is defined as
\begin{equation}
    \frac{1}{2}\norm{\Phi_1 - \Phi_2}_{\diamond,\bar{n}} \coloneqq 
    \sup_{\substack{\rho_{AA'}, \; \Tr [\rho_{AA'}(\hat{N}_A\otimes \id_{A'})]\le \bar{n}}} 
    \, \frac{1}{2}\norm{\Id_{A}\otimes \Phi_1\left(\rho_{AA'}\right) - \Id_{A}\otimes \Phi_2\left(\rho_{AA'}\right)}_{1},
\end{equation}
where the supremum is restricted to states whose mean photon number at the channel input system is bounded by $\bar{n}$. The operational meaning remains essentially the same: if one has access to a single use of an unknown channel, promised to be either $\Phi_1$ or $\Phi_2$ with equal prior probability $1/2$, and if there is an input energy constraint $\bar{n}$, then the optimal probability of correctly identifying the channel is 
\begin{equation}
    \frac{1}{2}\left(1 + \frac{1}{2}\norm{\Phi_1 - \Phi_2}_{\diamond,\bar{n}}\right).    
\end{equation}
This corresponds to the optimal distinguishing probability among all protocols where one prepares an input state (possibly entangled with an auxiliary system) whose input mean photon number is bounded by $\bar{n}$, sends it through the unknown channel, and performs a measurement. In conclusion, throughout our tomography analysis, we will use the energy-constrained diamond distance as the measure of error in the estimated channel.

\subsection{Upper bound on the energy-constrained diamond distance between Gaussian unitaries}
The energy-constrained diamond norm admits quantitative continuity bounds for Gaussian unitary channels, two of which 
we recall from~\cite{EC-diamond} and use later in our analysis. 

\begin{proposition}[(Bounds for displacement channels, see e.g.~{\cite[Eq.~(3)]{EC-diamond}})]
\label{prop:displacement_bound}
Let $\mathbf{r}_1,\mathbf{r}_2\in\R^{2m}$, and let $\mathcal{D}_{\mathbf{r}_1},\mathcal{D}_{\mathbf{r}_2}$ denote the displacement channels defined by $\mathcal{D}_{\mathbf{r}_1}(\rho) \coloneqq D_{\mathbf{r}_1} \rho D_{\mathbf{r}_1}^\dagger$ and $\mathcal{D}_{\mathbf{r}_2}(\rho) \coloneqq D_{\mathbf{r}_2} \rho D_{\mathbf{r}_2}^\dagger$, for all states $\rho$. Then, for every mean photon-number $\bar n \ge 0$,
\begin{align}
    \frac{1}{2}\norm{\mathcal{D}_{\mathbf{r}_1} - \mathcal{D}_{\mathbf{r}_2}}_{\diamond,\bar n} \le
    \sin\left(\min\left\{\frac{\bigl(\sqrt{\bar n}+\sqrt{\bar n+1}\bigr)}{\sqrt{2}}\cdot\norm{\mathbf{r}_1-\mathbf{r}_2}_2, \; \frac{\pi}{2}\right\}\right).
\end{align}
\end{proposition}

\begin{proposition}[(Bounds for symplectic Gaussian unitaries, see e.g.~{\cite[Eq.~(4)]{EC-diamond}})]
\label{thm:symplectic_bound}
Let $S_1,S_2 \in \mathrm{Sp}_{2m}(\R)$, and let $\mathcal{U}_{S_1}, \mathcal{U}_{S_2}$ denote the Gaussian unitary channels defined by 
$\mathcal{U}_{S_1}(\rho) \coloneqq U_{S_1} \rho U_{S_1}^\dagger$ and 
$\mathcal{U}_{S_2}(\rho) \coloneqq U_{S_2} \rho U_{S_2}^\dagger$, 
for all states $\rho$. Then, for every mean photon-number $\bar n \ge 0$, 
\begin{align}
    \frac{1}{2} \norm{\mathcal{U}_{S_1} - \mathcal{U}_{S_2}}_{\diamond,\bar n}
    \le
    \sqrt{\bigl(\sqrt{6}+\sqrt{10}+5\sqrt{2m}\bigr)(\bar n+1)}
    g\bigl(\norm{S_2^{-1}S_1}_\infty\bigr)
    \sqrt{\norm{S_2^{-1}S_1 - \id}_2},
\end{align}
where $g(x) \coloneqq \sqrt{{\pi}/{(x+1)}} + \sqrt{2x}$.
\end{proposition}

\section{The Gaussian unitary tomography problem}\label{sec:problem-define}
In this section we formalize the tomography task for Gaussian unitaries. The setting is the following: given access to queries of an unknown Gaussian unitary, the objective is to estimate it with performance guarantees measured in terms of the energy-constrained diamond distance. 

Recall that any Gaussian unitary can be expressed as $G_{\mathbf{r},S} \coloneqq D_{\mathbf{r}} U_S$, where $\mathbf{r}$ is the displacement vector and $S$ is the symplectic matrix. It is therefore natural that the tomography task depends on parameters associated with these quantities. 
In total, we will work with three parameters: $\bar n,\bar{n}_{\mathrm{in}} , z$, described as follows.
\begin{itemize}
    \item \textbf{Photon-number constraint in the error metric ($\bar n$).}  
    The primary notion of distance to quantify the error incurred in tomography will be the energy-constrained diamond distance 
    $\norm{\cdot}_{\diamond,\bar n}$, defined with respect to a bound $\bar n$ on the mean photon number.

    \item \textbf{Photon-number constraint at the input of the unknown unitary ($\bar{n}_{\mathrm{in}}$).} We assume that only input states with a mean photon number not exceeding $\bar{n}_{\mathrm{in}}$ can be fed into the unknown Gaussian unitary.

    \item \textbf{Squeezing bound ($z$).}  
    Similarly, we impose an upper bound on the operator norm of the symplectic matrix, $\norm{S}_\infty \le z$. The parameter $z$ quantifies the maximum squeezing allowed in the unknown Gaussian unitary. Indeed, by the Euler decomposition of Gaussian unitaries~\cite{BUCCO}, the operator norm $\norm{S}_\infty$ can be interpreted as the maximum squeezing parameter.
\end{itemize}

We can now state the tomography problem formally.

\begin{problem}[(Tomography of Gaussian unitaries)]\label{problem:tomography}
Let $m \in \mathds{N}$, $N_\mathrm{tot} \in \mathds{N}$, $z \ge 1$, $\bar n > 0$, $\bar n_{\mathrm{in}} > 0$, 
$\varepsilon \in (0,1)$, and $\delta \in (0,1)$ be known parameters.  Design a quantum algorithm that
\begin{itemize}
    \item \textbf{\emph{Given:}} black-box access to an unknown $m$-mode Gaussian unitary 
    $G_{\mathbf r,S} = D_{\mathbf r} U_S$, where 
    $S \in \mathrm{Sp}_{2m}(\R)$ satisfies $\|S\|_\infty \le z$;
    
    \item \textbf{\emph{Using:}} at most $N_{\mathrm{tot}}$ queries to $G_{\mathbf r,S}$ 
    and only input states with mean photon number at most $\bar n_{\mathrm{in}}$;
    
    \item \textbf{\emph{Outputs:}} estimators $\tilde{\mathbf r} \in \R^{2m}$ and $\tilde S \in \mathrm{Sp}_{2m}(\R)$; the corresponding Gaussian unitary channel $\tilde{\mathcal G} \coloneqq \mathcal D_{\tilde{\mathbf r}} \circ \mathcal U_{\tilde S}$ approximates the true channel $\mathcal G = \mathcal D_{\mathbf r} \circ \mathcal U_S$ and satisfies
    \begin{equation}
        \Pr\left[
            \frac{1}{2}\norm{\tilde{\mathcal{G}} - \mathcal{G}}_{\diamond,\bar n} 
            \le \varepsilon 
        \right] \ge 1-\delta,
    \end{equation}
    where the diamond norm is taken with respect to the mean photon number constraint $\bar n$.
\end{itemize}
\end{problem}

To assess the efficiency of a tomography task for channels, we introduce the notion of \emph{query complexity}. Query complexity is defined as the minimum number of queries $N_{\mathrm{tot}}$ to the unknown channel required to solve the tomography problem. Naturally, it depends on all the parameters that specify the problem. If the query complexity scales polynomially with the parameters defined in~\cref{problem:tomography} (in particular, the number of modes $m$ representing the system size, the energy constraint $\bar n$, the input mean photon number $\bar n_{\mathrm{in}}$, and the squeezing bound $z$), the tomography problem is said to be \emph{efficient}; otherwise, it is deemed \emph{inefficient}. In this work, we show that tomography of Gaussian unitaries is efficient, with query complexity growing polynomially in these parameters.

\section{Learning the symplectic component}
First, we address the algorithm for learning the symplectic component $S$ of a Gaussian unitary $G_{\mathbf r,S} = D_{\mathbf r}U_S$. For any coherent state $\ket{\mathbf m}$ with mean vector $\mathbf m \in \R^{2m}$, the action of a Gaussian unitary transforms its first and second moments as 
\begin{align}
    \mathbf m(G_{\mathbf r, S}\ket{\mathbf m}) &= \mathbf r + S\mathbf m, \\
    V(G_{\mathbf r, S}\ket{\mathbf m}) &= S \id S^{\top} = SS^{\top}.
\end{align}
The measurement scheme considered here is \emph{heterodyne detection}, which provides simultaneous noisy information about 
both quadratures of each mode (see~\cref{def:hetero} for a formal definition). Applying heterodyne detection to 
the Gaussian state $G_{\mathbf r, S}\ket{\mathbf m}$ yields
\begin{align}
    \mathsf{Heterodyne}(G_{\mathbf r, S}\ket{\mathbf m}) 
    \sim \mathcal N\left(\mathbf r + S\mathbf m, \; \frac{SS^{\top}+\id}{2}\right).
    \label{eq:heterodyne-distribution}
\end{align}
In words, this means that heterodyne detection produces classical samples centered at $\mathbf r + S\mathbf m$, 
with Gaussian fluctuations whose covariance is given by half the sum of $SS^\top$ and the identity.

Choosing the vacuum state $\ket{0}$ as input (with zero first moment) produces heterodyne outcomes
\begin{equation}
    Y_0 \sim \mathcal N\left(\mathbf r, \; \frac{SS^{\top}+\id}{2}\right).
\end{equation}
For columnwise probing, fix $\eta > 0$ and let $e_i \in \R^{2m}$ be the vector $(e_i)_j = \delta_{ij}$.  If we denote by $S_i$ the $i$-th column of $S$ (i.e.~$S_i\coloneqq Se_i$), then injecting the coherent state $\ket{\eta e_i}$ yields heterodyne outcomes
\begin{equation}
    Y_i \sim \mathcal N\left(\mathbf r + \eta S_i, \; \frac{SS^{\top}+\id}{2}\right).
\end{equation}
The scaled differences $(Y_i - Y_0)/\eta$ therefore provide unbiased estimators of the columns $S_i$, whose detailed statistical analysis will be carried out in ~\cref{sec:vac-share}. Stacking these column estimators produces a matrix estimate $\hat S$ of $S$, which in general may not be exactly symplectic. 
A nearby valid symplectic matrix $\tilde S$ can be obtained by the projection procedure described in~\cref{sec:close-sym}. 

\subsection{Symplectic learning with vacuum-shared inputs}\label{sec:vac-share}
In this section, we analyze the query complexity of learning the symplectic part $S$ via the method outlined above, which from now on we will refer as \emph{symplectic learning with vacuum-shared inputs}.

\begin{lemma}[(Query complexity for learning the symplectic part with vacuum-shared inputs)]
\label{thm:learnS-again}
Let $G_{\mathbf r, S} = D_{\mathbf r}U_S$ be a Gaussian unitary on $m$ bosonic modes with symplectic matrix $S \in \mathrm{Sp}_{2m}(\R)$.  
Fix the heterodyne measurement model for coherent input probes $\ket{\mathbf m}$, for which the outcome is a random vector
\begin{align}
Y \sim \mathcal N\left(\mathbf r + S\mathbf m,\, \Sigma\right),
\qquad
\Sigma \coloneqq \frac{SS^{\top}+\id}{2}.
\end{align}
In particular, let $Y_0$ denote the outcome distribution for the vacuum probe $\ket{0}$, and $Y_i$ the outcome distribution for the input coherent state $\ket{\eta e_i}$ with $\eta>0$ and $i \in [2m]$.  
From $N_S$ independent heterodyne samples of each probe, we form the empirical means
\begin{equation}
    \bar Y_0 \coloneqq \frac{1}{N_S}\sum_{k=1}^{N_S} Y_0^{(k)}, 
    \qquad
    \bar Y_i \coloneqq \frac{1}{N_S}\sum_{k=1}^{N_S} Y_i^{(k)}.    
\end{equation}
We then define the estimators
\begin{equation}
    \hat S_i \coloneqq \frac{\bar Y_i - \bar Y_0}{\eta}, 
    \qquad
    \hat S \coloneqq [\hat S_1, \dots, \hat S_{2m}].    
\end{equation}
Then, for every $\varepsilon > 0$ and $\delta \in (0,1)$, if
\begin{align}
\label{eq:NS-correct}
N_S \ge \frac{4m\|S\|_{\infty}^2 \bigl(\sqrt{2m}+\sqrt{2\log(2m/\delta)}\bigr)^{2}}{\eta^2\varepsilon^2},
\end{align}
we have 
\begin{equation}
    \Pr\left[\|\hat S - S\|_{\infty} \le \varepsilon \right] \ge 1-\delta.
\end{equation}
In particular, the total number of queries to the unitary $G_{\mathbf r, S}$ is $(2m+1)N_S$.
\end{lemma}

\begin{proof}
Let $n \coloneqq 2m$. Since $SS^{\top} \succcurlyeq 0$ and $I \succ 0$, we have $\Sigma \succ 0$, and hence $\Sigma^{\pm 1/2}$ are well defined.

\medskip
\noindent\emph{Step 1 (Columnwise distribution).}
For each $i \in [n]$, we obtain
\begin{align}
\bar Y_0 & \sim \mathcal N\left(\mathbf r, \, \frac{\Sigma}{N_S}\right), \\
\bar Y_i & \sim \mathcal N\left(\mathbf r+\eta S_i, \, \frac{\Sigma}{N_S}\right),
\end{align}
with the two batches independent. Consequently,
\begin{align}
\hat S_i - S_i
= \frac{1}{\eta} \cdot \left[\bar Y_i - \bigl(\mathbf r+\eta S_i\bigr) - \bigl(\bar Y_0-\mathbf r\bigr)\right]
\sim \mathcal N\left(0,\ \frac{2\Sigma}{\eta^2 N_S}\right). \label{eq:vac-col-dependent}
\end{align}

\medskip
\noindent\emph{Step 2 (Columnwise $\ell_2$ control via~\cref{lem:estrelative}).}
Applying~\cref{lem:estrelative} to the $N_S$ i.i.d.\ samples defining $\bar Y_i$ and $\bar Y_0$, we obtain that for any $\delta_i \in (0,1)$,
\begin{align}
\Pr\left[ \|\hat S_i - S_i\|_2  \le 
\sqrt{\frac{2\|\Sigma\|_\infty}{\eta^2 N_S}}
\cdot \bigl(\sqrt{n} + \sqrt{2\log(1/{\delta_i})}\bigr) \right] \ge 1-\delta_i.
\end{align}
Choosing $\delta_i = \delta/n$ and applying the union bound over $i \in [n]$ yields, for all $i \in [n]$,
\begin{align}
\label{eq:per-column}
\Pr \left[ \|\hat S_i - S_i\|_2 \le \sqrt{\frac{2\|\Sigma\|_\infty}{\eta^2 N_S}} \cdot
\bigl(\sqrt{n} + \sqrt{2\log(n/\delta)}\bigr), \,\, \forall i\in[n] \right]  \ge 1-\delta.
\end{align}

\medskip
\noindent\emph{Step 3 (From columns to operator norm).}
Let $E \coloneqq \hat S - S = [e_1,\ldots,e_n]$ with $e_i = \hat S_i - S_i$. Then
\begin{align}
\|\hat S - S\|_\infty & \le\ \|E\|_2 \nonumber \\
& = \left(\sum_{i=1}^n \|e_i\|_2^2\right)^{1/2} \nonumber \\
& \le \sqrt{n}\cdot\sqrt{\frac{2\|\Sigma\|_\infty}{\eta^2 N_S}} \cdot
\bigl(\sqrt{n} + \sqrt{2\log(n/\delta)}\bigr)
\end{align}
Combining this with \cref{eq:per-column}, we obtain
\begin{align}
\Pr \left[ \|\hat S - S\|_\infty \le \sqrt{\frac{2\|\Sigma\|_\infty}{\eta^2 N_S}}\cdot\sqrt{n}\cdot\bigl(\sqrt{n}+\sqrt{2\log(n/\delta)}\bigr)\right] \ge 1-\delta.
\end{align}

\medskip
\noindent\emph{Step 4 (Solve for $N_S$ and upper bound $\|\Sigma\|_\infty$).}
To ensure that $\|\hat S - S\|_\infty \le \varepsilon$, it suffices that
\begin{align}
N_S \ge \frac{2\|\Sigma\|_\infty}{\eta^2\varepsilon^2}\cdot
n\bigl(\sqrt{n}+\sqrt{2\log(n/\delta)}\bigr)^{2}.
\end{align}
Since $S$ is symplectic, we have $\|S\|_\infty \ge 1$ and
\begin{align}
\|\Sigma\|_\infty \le \frac12 \big(\|S\|_\infty^2+1\big) \le \|S\|_\infty^2,
\end{align}
which yields \cref{eq:NS-correct}. Substituting $n=2m$ completes the proof. By construction, the total number of queries is $(2m+1)N_S$.
\end{proof}

In the scheme of~\cref{thm:learnS-again}, letting $\eta \to +\infty$ allows one to take $N_S = 1$ for any fixed target accuracy. Consequently, the total number of queries reduces to $(2m+1)N_S = 2m+1$. For finite values of $\eta$, however, an alternative design may lead to a more favorable dependence on the sample complexity $N_S$. We therefore proceed to present such a design in the following subsection, which provides a complementary strategy for learning the symplectic part.

\subsection{Symplectic learning with symmetric probes}\label{sec:sym-probe}
In this section, we present an alternative algorithm for learning the symplectic component. The approach makes use of symmetric probes $\ket{\pm \eta e_i}$; accordingly, we will refer to it as \emph{symplectic learning with symmetric probes}. The key distinction from the vacuum-shared input case is that, in the latter, the columnwise samples in~\cref{eq:vac-col-dependent} are not independent, and thus the operator norm bound is obtained via a union bound over the individual columnwise estimators. By contrast, when symmetric probes are used, all columnwise samples become independent, which allows us to apply the Gaussian operator-norm tail bound. This is formalized in the following proposition.

\begin{lemma}[(Query complexity for learning the symplectic part with symmetric probes)]
\label{lem:pm-design}
Let $G_{\mathbf r, S} = D_{\mathbf r}U_S$ be a Gaussian unitary on $m$ bosonic modes with symplectic matrix $S \in \mathrm{Sp}_{2m}(\R)$.  
For each $i \in [2m]$, let $Y_{i,+}$ and $Y_{i,-}$ denote the outcome distributions corresponding to the coherent input states $\ket{+\eta e_i}$ and $\ket{-\eta e_i}$, respectively. From $N_S$ independent heterodyne samples of each probe, we form the empirical means 
\begin{align}
\bar Y_i^{(+)} \coloneqq \frac{1}{N_S}\sum_{k=1}^{N_S} Y_{i,+}^{(k)},
\qquad
\bar Y_i^{(-)} \coloneqq \frac{1}{N_S}\sum_{k=1}^{N_S} Y_{i,-}^{(k)}.
\end{align}
We then define the column estimators and matrix estimator as
\begin{align}
\hat S_i \coloneqq \frac{\bar Y_i^{(+)}-\bar Y_i^{(-)}}{2\eta},
\qquad
\hat S \coloneqq [\hat S_1,\ldots,\hat S_n].
\end{align}
Then, for every $\varepsilon > 0$ and $\delta \in (0,1)$, if
\begin{align}
\label{eq:NS-pm}
N_S \ge \frac{\|S\|_\infty^2\bigl(2\sqrt{2m}+\sqrt{2\log(1/\delta)}\bigr)^{2}}{2\eta^2\varepsilon^2},
\end{align}
we have
\begin{align}
\Pr\left[ \|\hat S - S\|_\infty \le \varepsilon \right] \ge\ 1-\delta.
\end{align}
In particular, the total number of queries to the unitary $G_{\mathbf r, S}$ is $4mN_S$.
\end{lemma}

\begin{proof}
Let $n \coloneqq 2m$. Since $SS^{\top} \succcurlyeq 0$ and $I \succ 0$, we have $\Sigma \succ 0$, and hence $\Sigma^{\pm 1/2}$ are well defined.

\medskip
\noindent\emph{Step 1 (Columnwise distribution).}
For each $i \in [n]$, we obtain
\begin{align}
\bar Y_i^{(+)} & \sim \mathcal N\left(\mathbf r+\eta S_i, \, \frac{\Sigma}{N_S}\right), \\
\bar Y_i^{(-)} & \sim \mathcal N\left(\mathbf r-\eta S_i, \, \frac{\Sigma}{N_S}\right),
\end{align}
with the two batches and disjoint across $i$ independent. Consequently,
\begin{align}
\hat S_i - S_i
= \frac{1}{2\eta}\left[\bigl(\bar Y_i^{(+)} - (\mathbf r+\eta S_i) \bigr) - \bigl(\bar Y_i^{(-)}- (\mathbf r-\eta S_i)\bigr)\right]
\sim \mathcal N\left(0, \, \frac{\Sigma}{2\eta^2 N_S}\right).
\end{align}
Moreover, the vectors $\{\hat S_i - S_i\}_{i=1}^n$ are mutually independent, since each is constructed from a separate pair of probe data sets.

\medskip
\noindent
\emph{Step 2 (Reduce to a standard Gaussian matrix).}
Define the random matrix
\begin{align}
X  \coloneqq \sqrt{2\eta^2 N_S}\, \Sigma^{-1/2} (\hat S - S)  \in \R^{n\times n}.
\end{align}
Fix any column index $i$. By Step~1 and $\Sigma\succ 0$,
\begin{align}
X_i = \sqrt{2\eta^2 N_S}\, \Sigma^{-1/2}(\hat S_i-S_i) \sim \mathcal N (0,\, \id).
\end{align}
Thus each column $X_i$ is a standard Gaussian vector in $\R^n$, and the columns are independent. Therefore all entries of $X$ are independent $\mathcal N(0,1)$ random variables, i.e., $X$ is a standard Gaussian matrix. Rearranging gives
\begin{align}
\hat S - S = \frac{\Sigma^{1/2}X}{\sqrt{2\eta^2 N_S}}.
\end{align}
By definition of the operator norm,
\begin{align}
\norm{\hat S - S}_{\infty} &= \frac{\norm{\Sigma^{1/2}X}_{\infty}}{\sqrt{2\eta^2 N_S}} \nonumber \\
& \le \frac{\norm{\Sigma^{1/2}}_{\infty}\ \norm{X}_{\infty}}{\sqrt{2\eta^2 N_S}} \nonumber \\
& = \sqrt{\frac{\norm{\Sigma}_{\infty}}{2\eta^2 N_S}} \, \norm{X}_{\infty}. \label{eq:deterministic-pm}
\end{align}

\medskip
\noindent
\emph{Step 3 (Concentration for $\norm{X}_\infty$).}
By the Gaussian operator-norm tail bound of~\cref{lem:gauss-opnorm-infty}, for all $t \ge 0$,
\begin{align}
\Pr\left[\|X\|_{\infty} \ge 2\sqrt{n}+t\right] \le e^{-t^2/2}.
\end{align}
Equivalently, for $\delta\in(0,1)$, taking $t=\sqrt{2\log(1/\delta)}$,
\begin{align}
\Pr\left[\|X\|_\infty \le 2\sqrt{n}+\sqrt{2\log(1/\delta)}\right] \ge 1-\delta.
\end{align}
Plugging this event into~\cref{eq:deterministic-pm} yields
\begin{align}
\Pr\left[\norm{\hat S - S}_{\infty}
\le
\sqrt{\frac{\|\Sigma\|_{\infty}}{2\,\eta^2 N_S}}\cdot\bigl(2\sqrt{n}+\sqrt{2\log(1/\delta)}\bigr)\right] \ge 1-\delta.
\end{align}

\medskip
\noindent
\emph{Step 4 (Final query complexity).}
To guarantee that $\|\hat S - S\|_{\infty} \le \varepsilon$, it suffices that
\begin{align}
N_S \ge \frac{\|\Sigma\|_\infty \bigl(2\sqrt{n}+\sqrt{2\log(1/\delta)}\bigr)^2}{2\eta^2\varepsilon^2}.
\end{align}
Since $S$ is symplectic, we have $\norm{S}_{\infty} \ge 1$, and hence
\begin{align}
    \norm{\Sigma}_{\infty} \le \frac{1}{2}(\norm{S}_{\infty}^2 + 1) \le \norm{S}_{\infty}^2,
\end{align}
which yields the bound in~\cref{eq:NS-pm}. Substituting $n=2m$ gives the claimed expression. Finally, by construction, the total number of queries is $4mN_S$.
\end{proof}

In the scheme of~\cref{lem:pm-design}, taking the limit $\eta \to +\infty$ allows one to set $N_S = 1$, which results in a total query complexity of $4mN_S = 4m$. For finite values of $\eta$, however, the symmetric-probe design may yield a more favorable dependence on the sample complexity $N_S$.

\subsection{Regularization of learned symplectic matrices}\label{sec:close-sym}
The estimators $\hat{S}$ obtained from the procedures in~\cref{sec:vac-share,sec:sym-probe} are not guaranteed to be symplectic. Therefore, it is necessary to design an efficient algorithm that approximates such an estimator by a nearby valid symplectic matrix. 

Consider any $\hat{S} \in \mathds{R}^{2m \times 2m}$ for which there exists a matrix $S \in \mathrm{Sp}_{2m}(\mathds{R})$ satisfying $\norm{\hat{S} - S}_\infty \leq \varepsilon$. We show how to round $\hat{S}$ to a matrix $\tilde{S} \in \mathrm{Sp}_{2m}(\mathds{R})$ such that $\norm{\tilde{S} - S}_\infty \leq 9\squeezingparameter^2 \varepsilon$, provided that $\varepsilon$ is sufficiently small. At a high level, our approach is based on symplectic analogues of polar decompositions, as discussed in~\cite{teretenkov2020symplectic}. Define
\begin{equation}
T \coloneqq -\Omega \hat{S}^\top \Omega \hat{S}.
\end{equation}
Observe that if $\hat{S}$ were symplectic, then $T$ would equal the identity matrix. 
We will argue that the principal root $Q = \sqrt{T}$ serves as a multiplicative correction to $\hat{S}$. 
Concretely, the rounding procedure sets
\begin{equation}
\tilde{S} = \hat{S} Q^{-1}.
\end{equation}

We motivate this rounding procedure by showing that if $Q = \sqrt{T}$ exists and is well-defined then $\tilde{S}$ is indeed symplectic.

\begin{proposition}[(Symplecticity of the rounded matrix)]
    Suppose that $Q = \sqrt{T}$ exists and is well-defined. Then $\tilde{S}\in \mathrm{Sp}_{2m}(\mathds{R})$.
\end{proposition}
\begin{proof}
Recall $T \coloneqq -\Omega \hat S^\top \Omega \hat S$ and $Q \coloneqq \sqrt{T}$ (the principal square root), so that $Q^2=T$. We first observe that
\begin{align}
(Q^2)^\top = T^\top & = (-\Omega \hat S^\top \Omega \hat S)^\top \nonumber \\
& = - \hat S^\top \Omega \hat S \Omega \nonumber \\
& = \Omega\,(-\Omega \hat S^\top \Omega \hat S)\Omega^{-1} \nonumber \\
& = \Omega Q^2 \Omega^{-1}, 
\end{align}
where we used $\Omega^\top=-\Omega$ and $\Omega^{-1}=-\Omega$.
Hence both $Q^\top$ and $\Omega Q \Omega^{-1}$ are square roots of $T^\top$.
Moreover, the principal square root is unique and is preserved by transpose and similarity, so $Q^\top = (T^\top)^{1/2} = \Omega Q \Omega^{-1}$. Using this identity, we compute
\begin{align}
\tilde{S}^\top \Omega \tilde{S} &= ({Q^{-1}})^\top \hat{S}^\top \Omega \hat{S} Q^{-1} \nonumber \\
&=\Omega Q^{-1} \Omega^{-1} \hat{S}^{\top}\Omega \hat{S} Q^{-1} \nonumber \\
&= \Omega Q^{-1}(- \Omega \hat{S}^{\top}\Omega \hat{S})Q^{-1} \nonumber \\
&= \Omega Q^{-1} Q^2 Q^{-1} = \Omega.
\end{align}
Hence $\tilde S \in \mathrm{Sp}_{2m}(\mathds{R})$.
\end{proof}

The remainder of the section is dedicated to showing that if $\hat{S}$ is sufficiently close to a symplectic matrix $S$ then $\tilde{S}$ is close to $S$ as well. The first steps towards this involve continuity properties of the principal root function.

In order for $\tilde{S}$ to exist, the principal root of $T$ must be well defined.
We show that if a matrix $A$ is sufficiently close to the identity in operator norm (a property we later demonstrate holds for $T$), then the principal root of $A$ is well-defined.
\begin{proposition}[(Well-definedness of the principal square root)]
\label{prop:root-well-defined}
    Let $A$ be such that $\norm{A - I}_{\infty} < 1$. Then the principal square root $\sqrt{A}$ exists and is well-defined.
\end{proposition}
\begin{proof} 
By~\cref{thm:principal-root}, the principal root of $A$ exists and is well-defined if and only if all of the eigenvalues of $A$ have no negative real component. Let $\lambda = \alpha + \beta i$ be an eigenvalue of $A$, where $\alpha,\beta \in \mathds{R}$. Since $A$ commutes with $I$ (that is, they are simultaneously diagonalizable), we have 
\begin{equation}
    \abs{\lambda - 1} = \sqrt{(1 - \alpha)^2 + b^2} < 1.
\end{equation}
This inequality implies $\abs{\alpha - 1} < 1$, hence $\alpha > 0$. 
Since this holds for all eigenvalues of $A$, all eigenvalues have positive real part, and therefore the principal square root $\sqrt{A}$ exists and is well-defined.
\end{proof}

Now we establish a continuity property of the matrix square root in a neighborhood of the identity.
\begin{proposition}[(Lipschitz continuity of the matrix square root near the identity)]\label{prop:root-lipschitz}
    Consider the open set $\mathcal{S} = \{A \in \mathds{R}^{n \times n} : \norm{A - I}_{\infty} < 1/2\}$. 
    For all $A \in \mathcal{S}$, 
    \begin{equation}
        \NORM{\sqrt{A} - \sqrt{I}}_\infty \leq (2-\sqrt{2}) \cdot \norm{A - I}_\infty.
    \end{equation}
\end{proposition}
\begin{proof}

Let $A = I + B$ with $\norm{B}_{\infty} < 1/2$. 
By \cref{thm:root-integral-rep}, whenever $A$ has no eigenvalues in the closed left half-plane of the complex plane, it admits the following integral representation:
\begin{align}
    \sqrt{A} - I &= \sqrt{I + B} - \sqrt{I} \nonumber \\
    &= \frac{2}{\pi} \int_{0}^\infty t^2 (t^2I + I + B)^{-1} B (t^2 I + I)^{-1} dt \nonumber \\
    &= \frac{2}{\pi} \int_{0}^\infty \frac{t^2}{t^2 + 1}(t^2I + A)^{-1} B dt.    
\end{align}
Applying the triangle inequality and submultiplicativity, we have
\begin{align}
    \norm{\sqrt{A} - I}_\infty & \leq \frac{2}{\pi} \int_{0}^\infty \frac{t^2}{t^2+1} \norm{(t^2I + A)^{-1} B}_\infty dt \nonumber \\
    & \leq \frac{2 \norm{B}_\infty}{\pi} \int_{0}^\infty \frac{t^2}{t^2+1} \norm{(t^2I + A)^{-1}}_\infty dt.    
\end{align}
We now bound the term $\norm{(t^2I+A)^{-1}}_\infty^2$. 
\begin{align}
    (t^2I+A)^{-1} & = ((t^2+1)I+B)^{-1} \nonumber \\
    &= (t^2 + 1)^{-1} \cdot (I + (t^2+1)^{-1} B)^{-1}.    
\end{align}
Consider the Neumann series
\begin{equation}
    \sum_{k=0}^\infty (- (t^2+1)^{-1} B)^k.    
\end{equation}
Since the radius of convergence is given by (see~\cref{fact:neumann})
\begin{equation}
    \frac{\norm{B}_\infty}{t^2+1} < \frac{1}{2(t^2+1)} < \frac12,
\end{equation}
this series converges to $(I + (t^2+1)^{-1} B)^{-1}$. Thus we can bound
\begin{align}
    \norm{(I + (t^2+1)^{-1} B)^{-1}}_\infty & \leq \sum_{k=0}^{\infty} \norm{(t^2+1)^{-1} B}_\infty^k \nonumber \\
    &\leq \sum_{k=0}^{\infty} \left(\frac{1}{2(t^2+1)}\right)^k \nonumber \\
    & = \frac{2(t^2+1)}{2t^2+1}.
\end{align}
As such, $\norm{(t^2I+A)^{-1}}_\infty^2 \leq {2}/(2t^2+1)$. Substituting this into the integral bound yields:
\begin{align}
    \NORM{\sqrt{A} - I}_\infty & \le \left(\frac{2}{\pi}\int_0^\infty \frac{2t^2}{(2t^2+1)}dt\right) \cdot \norm{B}_\infty \nonumber \\
    & = (2 - \sqrt{2}) \cdot \norm{A - I}_\infty.
\end{align}
This completes the derivation of the desired bound.
\end{proof}

We are now ready to prove that $\tilde{S}$ constitutes an effective rounding of $\hat{S}$ to the symplectic group, and is therefore a symplectic matrix that is close to $S$.

\begin{lemma}[(Error bound for symplectic rounding)]
Suppose $\norm{\hat{S} - S}_\infty \leq \varepsilon$, $\norm{S}_\infty \leq \squeezingparameter$, and $(2\squeezingparameter + 1) \varepsilon < 1/2$. Then $\norm{\tilde{S} - S}_\infty \leq 9\squeezingparameter^2 \varepsilon$.
\end{lemma}

\begin{proof}
First, we must argue that $\tilde{S}$ is well-defined. This follows from the existence, well-definedness, and non-singularity of $Q$. Indeed, 
    \begin{align}
        \norm{T - I}_\infty &= \norm{-\Omega \hat{S}^\top \Omega \hat{S} - I}_\infty \nonumber \\
        &= \norm{-\Omega \hat{S}^\top \Omega \hat{S} + \Omega S^\top \Omega S}_\infty  \nonumber \\
        &= \norm{\hat{S}^\top \Omega \hat{S} - S^\top \Omega S}_\infty & \text{(unitary invariance)} \nonumber \\
        &= \norm{\hat{S}^\top \Omega \hat{S} - \hat{S}^\top \Omega S}_\infty + \norm{\hat{S}^\top \Omega S - S^\top \Omega S}_\infty &\text{(triangle inequality)} \nonumber \\
        &= \norm{\hat{S}^\top \Omega}_\infty \cdot \norm{\hat{S} - S}_\infty + \norm{\hat{S}^\top - S^\top}_\infty \cdot \norm{\Omega S}_\infty &\text{(submultiplicativity)} \nonumber\\
        &\leq \bigl(\norm{\hat{S}^\top}_\infty + \norm{S}_\infty\bigr) \eps \nonumber\\
        &\leq (2\squeezingparameter + 1) \eps < 1/2.
    \end{align}
    Thus, by \cref{prop:root-well-defined}, $Q = \sqrt{T}$ exists and is well-defined. Moreover, $\abs{\det(Q)}^2 = \abs{\det(Q^2)} = \abs{\det(T)} > 0$, so $Q$ is invertible. Hence $\tilde S$ is well-defined, and we now show that it is close to $S$ in operator norm.    
    
    By the triangle inequality, we have 
    \begin{align}
        \norm{\tilde{S} - S}_\infty & \leq \norm{\tilde{S} - \hat{S}}_\infty + \norm{\hat{S} - S}_\infty \nonumber \\
        & \leq \norm{\tilde{S} - \hat{S}}_\infty + \varepsilon.
    \end{align}

    The remainder of the proof bounds $\norm{\tilde{S} - \hat{S}}_{\infty}$. 
    Observe that $\norm{\hat{S}}_\infty \leq \norm{S}_\infty + \norm{\hat{S} - S}_\infty \leq \squeezingparameter + \varepsilon$. Applying submultiplicativity of the operator norm twice, we have
    \begin{align}
        \norm{\tilde{S} - \hat{S}}_\infty & = \norm{\hat{S} Q^{-1} - \hat{S}}_\infty \nonumber \\
        & \leq \norm{\hat{S}}_\infty \cdot \norm{Q^{-1} - I}_\infty \nonumber \\
        & \leq (\squeezingparameter + \varepsilon) \cdot \norm{Q^{-1} - I}_\infty \nonumber \\
        & \leq (\squeezingparameter + \varepsilon) \cdot \norm{Q^{-1}}_\infty \cdot \norm{Q - I}_\infty.
    \end{align}
    Let $\alpha = 2 - \sqrt{2}$ for notational convenience. Since $\norm{T - I}_\infty < 1/2$, \cref{prop:root-lipschitz} gives
    \begin{equation}
        \norm{Q - I}_\infty \leq \alpha \norm{T - I}_\infty \leq \alpha(2\squeezingparameter + 1)\varepsilon.
    \end{equation}
    To bound $\norm{Q^{-1}}_{\infty}$, we use the Neumann series:
    \begin{equation}
        Q^{-1}=\sum_{k=0}^\infty (I - Q)^k.    
    \end{equation}
    Since $\norm{Q - I}_\infty \leq \alpha(2\squeezingparameter + 1)\varepsilon < 1/2$, the series converges, yielding
    \begin{align}
        \norm{Q^{-1}}_\infty & \leq \sum_{k=0}^{\infty} \norm{Q - I}_\infty^k \leq \sum_{k=0}^{\infty}(\alpha(2\squeezingparameter + 1)\varepsilon)^k \nonumber \\
        & = \left(1 - \alpha(2\squeezingparameter + 1)\varepsilon\right)^{-1} \leq 1 + 2 \alpha (2\squeezingparameter + 1) \varepsilon,    
    \end{align}
    where the last inequality follows from the fact that $\alpha(2\squeezingparameter + 1)\eps < 1/2$.
    Putting everything together,
    \begin{align}
        \norm{\tilde{S} - S}_\infty &\leq \norm{\tilde{S} - \hat{S}}_\infty + \norm{\hat{S} - S}_\infty \nonumber \\
        &\leq \norm{\tilde{S} - \hat{S}}_\infty + \varepsilon \nonumber \\
        &\leq (z + \varepsilon) \cdot \norm{Q^{-1}}_\infty \cdot\norm{Q - I}_\infty + \varepsilon \nonumber \\
        &\leq \alpha(z + 1)\left(1 + 2\alpha (2\squeezingparameter + 1) \varepsilon \right)(2 \squeezingparameter + 1)\varepsilon + \varepsilon \nonumber \\
        &\leq 2\alpha(z + 1)(2 \squeezingparameter + 1)\varepsilon + \varepsilon \nonumber \\
        &\leq 2\alpha(2\squeezingparameter^2 + 3\squeezingparameter + 1)\varepsilon + \varepsilon \leq 9 \squeezingparameter^2 \varepsilon.
    \end{align}
This completes the derivation of the desired bound.
\end{proof}

\subsection{Learning regularized symplectic matrices}
In the previous section, we described a procedure to round an approximate estimate of a symplectic matrix to an exactly symplectic one. We now discuss how to incorporate this result into the learning algorithm. The regularization result from before can be summarized as follows.

\begin{lemma}[(Symplectic regularization, see~\cref{sec:close-sym} for details)]
\label{lem:symplectic-regularization}
Let $S \in \mathrm{Sp}_{2m}(\R)$ be a symplectic matrix, and let $\hat S \in \R^{2m \times 2m}$ satisfy $\norm{\hat S - S}_\infty \le \varepsilon$ and $\left(2\norm{S}_{\infty}+1\right)\varepsilon < 1/2$. Define $T \coloneqq -\Omega \hat S^\top \Omega \hat S$, $Q \coloneqq \sqrt{T}$, and $\tilde S \coloneqq \hat S Q^{-1}$. Then $\tilde S \in \mathrm{Sp}_{2m}(\R)$ and $\norm{\tilde S - S}_{\infty} \le 9 \norm{S}_{\infty}^2 \varepsilon$. In particular, any matrix $\hat S$ that is $\varepsilon$-close to a symplectic matrix $S$ can be efficiently regularized to an exactly symplectic matrix $\tilde S$ that remains $\mathcal{O}(\norm{S}_{\infty}^2\varepsilon)$-close to $S$.
\end{lemma}

Consequently, the query complexity required to learn a regularized symplectic matrix can be expressed as follows.

\begin{proposition}[(Learning a regularized symplectic matrix with vacuum-shared inputs)]
\label{prop:regularized-S-vac}
Let $G_{\mathbf r, S} = D_{\mathbf r} U_S$ be a Gaussian unitary on $m$ bosonic modes with symplectic matrix $S \in \mathrm{Sp}_{2m}(\R)$ satisfying $\norm{S}_{\infty} \le z$. Fix accuracy $\tau \in (0,1)$ and failure probability $\delta \in (0,1)$. If the coherent-state estimation protocol of~\cref{thm:learnS-again} (with probes $\ket{0}$ and $\ket{\eta e_i}$) is followed by the symplectic regularization procedure of~\cref{lem:symplectic-regularization}, then whenever
\begin{align}
N_S \ge \frac{324mz^6\bigl(\sqrt{2m}+\sqrt{2\log(2m/\delta)}\bigr)^{2}}{\eta^2\tau^2},
\end{align}
heterodyne shots are used per probe, we have
\begin{equation}
\Pr\left[
   \|\tilde S - S\|_{\infty} \le \tau
\right] \ge 1-\delta.
\end{equation}
where $\tilde S \in \mathrm{Sp}_{2m}(\R)$. In particular, the total query complexity is $(2m+1)N_S$ uses of $G_{\mathbf r, S}$.
\end{proposition}

\begin{proof}
By~\cref{thm:learnS-again}, after $N_S$ heterodyne samples per probe the estimate $\hat S$ satisfies $\norm{\hat S - S}_{\infty} \le \varepsilon$ with probability at least $1-\delta$, provided
\begin{equation}
    N_S \ge  \frac{4mz^2\bigl(\sqrt{2m}+\sqrt{2\log(2m/\delta)}\bigr)^2}{\eta^2 \varepsilon^2}.
\end{equation}
Applying~\cref{lem:symplectic-regularization}, if $(2z+1)\varepsilon < 1/2$, then the regularized matrix $\tilde S$ is symplectic and obeys $\norm{\tilde S - S}_{\infty} \le 9z^2 \varepsilon$. Choosing $\varepsilon = \tau/(9z^2)$ ensures $\norm{\tilde S - S}_{\infty} \le \tau$. 

Since $z \ge 1$, the smallness condition $(2z+1)\varepsilon < 1/2$ is automatically satisfied for $\tau \in (0,1)$. 
Substituting $\varepsilon = \tau/(9z^2)$ into the query complexity bound gives the claimed bound. The total query complexity is $(2m+1)N_S$.
\end{proof}

\begin{proposition}[(Learning a regularized symplectic matrix with symmetric probes)]
\label{prop:regularized-S-sym}
Let $G_{\mathbf r, S} = D_{\mathbf r} U_S$ be a Gaussian unitary on $m$ bosonic modes with symplectic matrix $S \in \mathrm{Sp}_{2m}(\R)$ satisfying $\norm{S}_{\infty} \le z$. Fix accuracy $\tau \in (0,1)$ and failure probability $\delta \in (0,1)$. If the coherent-state estimation protocol of~\cref{lem:pm-design} (with probes $\ket{\pm\eta e_i}$) is followed by the symplectic regularization procedure of~\cref{lem:symplectic-regularization}, then whenever
\begin{align}
N_S \ge \frac{81z^6\bigl(2\sqrt{2m}+\sqrt{2\log(1/\delta)}\bigr)^{2}}{2\eta^2\tau^2},
\end{align}
heterodyne shots are used per probe, we have
\begin{equation}
\Pr\left[
   \|\tilde S - S\|_{\infty} \le \tau
\right] \ge 1-\delta.
\end{equation}
where $\tilde S \in \mathrm{Sp}_{2m}(\R)$. In particular, the total query complexity is $4mN_S$ uses of $G_{\mathbf r, S}$.
\end{proposition}

\begin{proof}
By~\cref{lem:pm-design}, after $N_S$ heterodyne samples per probe the estimate $\hat S$ satisfies $\norm{\hat S - S}_{\infty} \le \varepsilon$ with probability at least $1-\delta$, provided
\begin{equation}
    N_S \ge  \frac{z^2\bigl(2\sqrt{2m}+\sqrt{2\log(1/\delta)}\bigr)^{2}}{2\eta^2\varepsilon^2}.
\end{equation}
The subsequent proof then follows the same steps as in~\cref{prop:regularized-S-vac}, and by setting $\varepsilon = \tau/(9z^2)$ all the required conditions are satisfied. The total query complexity is $4mN_S$.
\end{proof}

\section{Learning the displacement component}
Next, we address the algorithm for learning the displacement component $\mathbf r$ of the Gaussian unitary $G_{\mathbf r,S} = D_{\mathbf r}U_S$. We present two versions of the learning algorithm. In~\cref{sec:dis-entangle,sec:algo-without-entanglement}, we describe and analyze the query complexities of two approaches: one based on two-mode squeezed states that requires entanglement with an auxiliary system, and another based on single-mode squeezed states that does not require such entanglement.

\subsection{Displacement learning with auxiliary-system entanglement}\label{sec:dis-entangle}
To estimate $\mathbf r$, we prepare a product two-mode-squeezed vacuum $\ket{\nu}^{\otimes m}=U_{S_\nu}\ket{0}^{\otimes 2m}$, apply first $U_{\tilde S^{-1}}$, then $G$, and then $U_{S_\nu}^\dagger$, where
\begin{equation}
    S_{\nu}\coloneqq
    \begin{pmatrix}
    \sqrt{\nu}I & \sqrt{\nu-1}Z\\[2pt]
    \sqrt{\nu-1}Z & \sqrt{\nu}I
    \end{pmatrix},\qquad
    Z\coloneqq\bigoplus_{i=1}^{m}
    \begin{pmatrix}1&0\\0&-1\end{pmatrix}.
\end{equation}

The resulting state after this sequence of operations has the following moments.

\begin{lemma}[(Moments of the output state of the protocol)]
\label{lem:moments}
Let $\Delta := \tilde S^{-1} S - \id$. 
Then, the first and second moments of the protocol output can be expressed as
\begin{align}
\mathbf m \bigl(U_{S_\nu}^\dagger G_{\mathbf r, S} U_{\tilde S^{-1}}\ket{\nu}^{\otimes m}\bigr)
    &= (\sqrt{\nu} \, \mathbf r,\; -\sqrt{\nu-1}Z\mathbf r), \\
V\bigl(U_{S_\nu}^\dagger G_{\mathbf r, S} U_{\tilde S^{-1}}\ket{\nu}^{\otimes m}\bigr) & = 
    \begin{pmatrix}
        A & C\\[2pt] C^{\top} & B
    \end{pmatrix},
\end{align}
where
\begin{align}
A &= \id + \nu(\Delta+\Delta^{\top}) + \nu(2\nu+1)\,\Delta\Delta^{\top},\\
B &= \id - (\nu-1)(\Delta+\Delta^{\top}) + (\nu+1)(2\nu+1)\,\Delta\Delta^{\top},\\
C &= \Bigl[-(2\nu-1)\sqrt{\nu(\nu-1)}\Delta\Delta^{\top}
    -\sqrt{\nu(\nu-1)}\Delta^{\top}
    +\sqrt{\nu(\nu-1)}\Delta\Bigr]Z.
\end{align}
\end{lemma}
\begin{proof}
    The proof follows from a direct block-matrix calculation of the combined symplectic transformation; details are deferred to~\cref{app:proof-moments}.
\end{proof}
If $\nu$ is large and $\norm{\Delta}_{\infty} =o(1/\nu)$ (i.e., the estimated symplectic $\tilde S$ is very close to the true symplectic $S$), then a few rounds of heterodyne detection on $U_{S_\nu}^\dagger G U_{\tilde S^{-1}} \ket{\nu}^{\otimes m}$ produce an estimate of the displacement $\mathbf r$ with error $o(1/\sqrt{\nu})$, with high probability. Then, before analyzing the query complexity for displacement learning, let us recall that, by \cref{eq_tmsv}, the tensor product of $n$ two-mode squeezed vacuum (TMSV) states can be expressed as 
\begin{equation}
    \ket{\nu}^{\otimes m} = U_{S_\nu}\ket{0}^{\otimes 2m},    
\end{equation}
namely as the action of the Gaussian unitary associated with the symplectic matrix $S_\nu$ on the $2m$-mode vacuum state $\ket{0}^{\otimes 2m}$, where
\begin{equation}
    S_{\nu} \coloneqq 
    \begin{pmatrix}
        \sqrt{\nu}\,\id & \sqrt{\nu-1}\,Z \\[2pt]
        \sqrt{\nu-1}\,Z & \sqrt{\nu}\,\id
    \end{pmatrix}, \qquad
    Z \coloneqq \bigoplus_{i=1}^n \begin{pmatrix} 1 & 0 \\ 0 & -1 \end{pmatrix}.
\end{equation}

\begin{lemma}[(Query complexity for learning the displacement with auxiliary-system entanglement)]
\label{thm:learn-r-nu-only}
Let $G_{\mathbf r, S} = D_{\mathbf r} U_S$ act on $m$ bosonic modes (phase-space dimension $2m$). 
Assume that an estimate $\tilde S$ of the symplectic part has been obtained and implemented via $U_{\tilde S^{-1}}$, and define $\Delta \coloneqq \tilde S^{-1}S - \id$.
Prepare $m$ copies of the two-mode squeezed vacuum $\ket{\nu}^{\otimes m}$, apply in order $U_{\tilde S^{-1}}$, then $G_{\mathbf r, S}$, and then $U_{S_\nu^{-1}}$, where
\begin{equation}
    S_{\nu}^{-1} = 
\begin{pmatrix}
\sqrt{\nu}\,\id & -\sqrt{\nu-1}\,Z \\[2pt]
-\sqrt{\nu-1}\,Z & \sqrt{\nu}\,\id
\end{pmatrix}. 
\end{equation}
The first $2m$ output outcome coordinates then have mean $\sqrt{\nu}\, \mathbf r$ and covariance
\begin{equation}
    \Sigma^{(1)}  = \frac{A+\id}{2}, \qquad 
    A = \id + \nu(\Delta+\Delta^{\top}) + \nu(2\nu+1)\Delta\Delta^{\top}.
\end{equation}
Performing heterodyne detection on these modes yields $Y^{(1)} \sim \mathcal N(\sqrt{\nu}\, \mathbf r, \,\Sigma^{(1)})$.
Repeating the experiment $N_r$ times gives independent samples  $Y^{(1)}_1,\dots,Y^{(1)}_{N_r}$, and we estimate
\begin{equation}
    \tilde{\mathbf r} \coloneqq \frac{\hat{\mu}^{(1)}}{\sqrt{\nu}}, \qquad 
    \hat{\mu}^{(1)} = \frac{1}{N_r}\sum_{k=1}^{N_r} Y^{(1)}_k.
\end{equation}
Then, for every $\varepsilon>0$ and $\delta\in(0,1)$, if
\begin{align}
N_r \ge \frac{\bigl(1 + \nu\norm{\Delta}_{\infty}+(3/2)\cdot(\nu\norm{\Delta}_{\infty})^2\bigr) \bigl(\sqrt{2m}+\sqrt{2\log(1/\delta)}\bigr)^2}{\nu \varepsilon^2},
\end{align}
we have 
\begin{equation}
    \Pr\left[\norm{\tilde{\mathbf r}-\mathbf r}_2 \le \varepsilon\right] \ge 1-\delta.
\end{equation}
\end{lemma}

\begin{proof}
Since $Y^{(1)}\sim\mathcal N(\sqrt{\nu}\,\mathbf r, \,\Sigma^{(1)})$, the empirical mean satisfies the Gaussian concentration inequality (see~\cref{lem:estrelative}):
\begin{equation}
    \Pr \left[\norm{\hat{\mu}^{(1)} - \sqrt{\nu} \, \mathbf r}_2  \le \frac{\chi_{2m,\delta}}{\sqrt{N_r}} \sqrt{\norm{\Sigma^{(1)}}_{\infty}}\right] \ge 1-\delta
\end{equation}
where $\chi_{2m,\delta}=\sqrt{2m}+\sqrt{2\log(1/\delta)}$.  
To bound $\norm{\Sigma^{(1)}}_{\infty}$, note that
\begin{align}
    \norm{A}_{\infty} & \le \left(2 + \frac{1}{\nu}\right) (\nu \norm{\Delta}_{\infty})^2 + 2(\nu\norm{\Delta}_{\infty}) + 1 \nonumber \\
    & \le 3\left(\nu \norm{\Delta}_{\infty}\right)^2 + 2\nu \norm{\Delta}_{\infty} + 1,    
\end{align}
since $\nu\ge1$. Therefore
\begin{align}
    \norm{\Sigma^{(1)}}_{\infty} & = \frac12 \norm{A+\id}_{\infty} \nonumber \\
    & \le 1 + \nu \norm{\Delta}_{\infty}+\frac32 (\nu\norm{\Delta}_{\infty})^2.
\end{align}
Finally,
\begin{align}
    \norm{\tilde{\mathbf r}-\mathbf r}_2 & = \frac{1}{\sqrt{\nu}} \norm{\hat{\mu}^{(1)} - \sqrt{\nu}\mathbf r}_2 \nonumber \\
    & \le \frac{\chi_{2m,\delta} \cdot \bigl(1+\nu \norm{\Delta}_{\infty} + (3/2)\cdot(\nu \norm{\Delta}_{\infty})^2\bigr)^{1/2}}{\sqrt{N_r}\sqrt{\nu}},
\end{align}
so requiring the right-hand side to be at most $\varepsilon$ gives the stated bound on $N_r$.
\end{proof}

\subsubsection{Experiment-friendly implementation without active squeezing}\label{sec:no-activesq}

In one stage of our protocol, we prepare copies of
\begin{align}
U_{S_\nu}^\dagger  G_{\mathbf r, S}  U_{\tilde{S}^{-1}} \ket{\nu}^{\otimes m},
\end{align}
and then perform heterodyne detection. At first glance, this appears to require
\begin{itemize}
    \item[(i)] the \emph{online} application of the \emph{active} Gaussian unitary $U_{S_\nu}^\dagger$ prior to measurement, and \item [(ii)] applying $U_{\tilde{S}^{-1}}$ to a squeezed input.
\end{itemize}
Active (general) Gaussian operations are experimentally demanding, whereas \emph{passive} Gaussian unitaries (those whose symplectic matrix is also orthogonal, i.e., in $\mathrm{Sp}_{2n}(\mathds{R}) \cap \mathrm{O}(2n)$, such as beam splitters and phase shifters) and \emph{input} (offline) squeezing via pre-prepared squeezed-vacuum ancillas are markedly more accessible experimentally. We now show that neither (i) nor (ii) requires online squeezing: this entire block can be implemented \emph{exactly} using only passive linear optics and input squeezed states, followed by complementary homodyne measurements and a linear post-processing of the recorded outcomes (see also~{\cite[Section E and Appendix A]{bittel2025energyindependenttomographygaussianstates}}).

\paragraph{Removing $U_{S_\nu}^\dagger$.}
The joint action “active Gaussian then heterodyne” can be realized \emph{passively} as follows.

\begin{proposition}[(Passive realization of heterodyne after an active Gaussian, see e.g.~{\cite[Appendix A and Proposition 15]{bittel2025energyindependenttomographygaussianstates}})]
\label{prop:passive-het}
Let $\rho$ be an $n$-mode Gaussian state and let $U_S$ be a Gaussian unitary with symplectic matrix $S\in\mathrm{Sp}_{2n}(\R)$. The heterodyne outcome distribution on $U_S\rho U_S^\dagger$ coincides with that of the following passive scheme:
\begin{enumerate}
\item Prepare an auxiliary squeezed vacuum $\sigma$ with covariance $S^{-1}S^{-\top}$ (i.e., the state $U_S^{-1}\ket{0}$).
\item Interfere $\rho$ and $\sigma$ on a balanced beam splitter, mode by mode.
\item Homodyne the position quadratures on one output arm and the momentum quadratures on the other; denote the outcome by $\mathbf q\coloneqq \sqrt{2}(\mathbf x, \mathbf p)\in\R^{2n}$.
\item Output the post-processed variable $S \mathbf q$.
\end{enumerate}
Then $S\mathbf r$ has the same Gaussian distribution as the heterodyne outcome on $U_S\rho U_S^\dagger$.
\end{proposition}

\noindent
Applying~\cref{prop:passive-het} with $S=S_\nu$ shows that the apparent online use of $U_{S_\nu}^\dagger$ can be dispensed with: the unknown signal remains untouched; one injects a fixed squeezed ancilla, uses only passive optics and complementary homodynes, and finally multiplies the recorded outcome by $S_\nu$.

\paragraph{Rewriting $U_{\tilde{S}^{-1}} \ket{\nu}^{\otimes m}$ using only offline squeezing.}
Note that
\begin{align}
U_{\tilde{S}^{-1}} \ket{\nu}^{\otimes m}
& = U_{\tilde{S}^{-1}} U_{S_\nu} \ket{0}^{\otimes m} \nonumber  \\
& = U_{\tilde{S}^{-1} S_\nu} \ket{0}^{\otimes m}.
\end{align}
By the Euler (Bloch–Messiah) decomposition,
\begin{align}
\tilde{S}^{-1} S_\nu = O_1  Z  O_2,
\end{align}
with $O_1,O_2\in\mathrm{Sp}_{2n}(\mathds{R})\cap\mathrm{O}(2n)$ passive and $Z=\mathrm{diag}(z_1,z_1^{-1},\dots,z_n,z_n^{-1})$ diagonal (single-mode squeezers). On the unitary level,
\begin{align}
U_{\tilde{S}^{-1} S_\nu}\ket{0}^{\otimes m}
& = U_{O_1} U_Z U_{O_2}\ket{0}^{\otimes m} \nonumber  \\
& = U_{O_1} U_Z \ket{0}^{\otimes m} \nonumber  \\
& = U_{O_1}\ket{Z},
\end{align}
since $U_{O_2}\ket{0}^{\otimes m}=\ket{0}^{\otimes m}$. Hence
\begin{align}
U_{\tilde{S}^{-1}} \ket{\nu}^{\otimes m} = U_{O_1}\ket{Z}.
\end{align}
Thus, $U_{\tilde{S}^{-1}}\ket{\nu}^{\otimes m}$ can be implemented by applying the passive interferometer $U_{O_1}$ to the offline-prepared squeezed state $\ket{Z}$.

Therefore, both apparent online uses of online squeezing are removed:
\begin{itemize}
\item The block “$U_{S_\nu}^\dagger$ + heterodyne’’ is implemented passively with a fixed squeezed ancilla, homodyne detection, and the linear post-processing $\mathbf{q}\mapsto S_\nu \mathbf{q}$ (\cref{prop:passive-het}).
\item The input $U_{\tilde{S}^{-1}} \ket{\nu}^{\otimes m}$ is realized as offline \emph{input} squeezing $\ket{Z}$ followed by a passive interferometer $U_{O_1}$.
\end{itemize}
Thus, this stage of the protocol uses only passive Gaussian unitaries, input (offline) squeezed states, and homodyne detection—\emph{no} online active squeezing on the unknown signal—while preserving the heterodyne outcome statistics and all ensuing guarantees of our analysis.

\subsection{Displacement learning without auxiliary-system entanglement}\label{sec:algo-without-entanglement}

Two-mode squeezed states exhibit entanglement shared between the input modes and auxiliary modes. We now present an alternative learning algorithm that requires no entanglement with an auxiliary system, which is instead based solely on single-mode squeezed states combined with homodyne detection.

\paragraph{Learning the momentum component of the displacement.} 
Fix $z_{\mathrm{in}} \ge 1$. This step proceeds as follows:
\begin{itemize}
    \item Prepare the pure Gaussian state $U_{\tilde{S}^{-1}}\ket{z_{\mathrm{in}}}^{\otimes m}$, where $\ket{z_{\text{in}}}^{\otimes m}$ is a tensor product of $m$ single-mode squeezed states 
    squeezed in the momentum quadrature by factor $z_{\text{in}}$. This state has zero first moment and covariance matrix
    \begin{equation}
        V(\ket{z_{\text{in}}})=\bigoplus_{i=1}^m \begin{pmatrix}
            z_{\text{in}} & 0 \\[4pt] 0 & z_{\text{in}}^{-1}
        \end{pmatrix}.
    \end{equation}
    Note that this step is experimentally feasible: by Euler decomposition, any pure Gaussian state can be implemented using only passive Gaussian unitaries and input squeezing, without requiring online squeezing (see~\cref{sec:no-activesq}).

    \item Apply the unknown Gaussian unitary $G_{\mathbf r, S}$, so that the resulting state is $G_{\mathbf r, S}U_{\tilde{S}^{-1}}\ket{z_{\text{in}}}^{\otimes m}$.
    
    \item Perform homodyne detection of the momentum quadrature.
\end{itemize}
The resulting state has first moment and covariance matrix
\begin{align}
    \mathbf{m} \bigl(G_{\mathbf r, S}U_{\tilde{S}^{-1}}\ket{z_{\text{in}}}^{\otimes m}\bigr) &= \mathbf{r}, \\ 
    V \bigl(G_{\mathbf r, S}U_{\tilde{S}^{-1}}\ket{z_{\text{in}}}^{\otimes m}\bigr) &= S \tilde{S}^{-1} V(\ket{z_{\text{in}}}) (S\tilde{S}^{-1})^\top \\
    & = (\Delta+\id) V(\ket{z_{\text{in}}}) (\Delta+\id)^\top,
\end{align}
where $\Delta \coloneqq S\tilde{S}^{-1}-\id$. Let $\mathbf{r}_p$ be the momentum component of $\mathbf{r}$, and for any $2m\times 2m$ matrix $A$ let $A_{pp}$ denote its momentum-momentum block.  
Then homodyne detection of the momentum quadrature yields a random outcome $Y_p$ distributed as
\begin{equation}
    Y_p \sim \mathcal{N}\left(\mathbf{r}_p, \, \frac{V_{pp}}{2}\right),
\end{equation}
with
\begin{equation}
    V_{pp} \coloneqq \bigl( (\Delta+\id) V(\ket{z_{\mathrm{in}}}) (\Delta+\id)^\top \big)_{pp}.
\end{equation}
Consequently, the covariance of $Y_p$ obeys
\begin{equation}\label{eq_op_normvpp}
    \NORM{\frac{V_{pp}}{2}}_\infty \le \frac{1}{2}\left(\frac{(1+\|\Delta\|_\infty)^2}{z_{\text{in}}}+z_{\text{in}}\|\Delta\|_\infty^2\right),
\end{equation}
Hence, in the regime $\norm{\Delta}_\infty = o(1/z_{\text{in}})$ (corresponding to an accurate estimate of the symplectic component), the measurement covariance vanishes as $z_{\mathrm{in}} \to \infty$, enabling efficient estimation of $\mathbf{r}_p$.

\paragraph{Learning the position component of the displacement.} 
Fix $z_{\text{in}} \ge 1$. This step proceeds analogously:
\begin{itemize}
    \item Prepare the pure Gaussian state $U_{\tilde{S}^{-1}}\ket{z_{\mathrm{in}}^{-1}}^{\otimes m}$, where $\ket{z_{\mathrm{in}}^{-1}}^{\otimes m}$ is a tensor product of $m$ single-mode squeezed states with position variance $z_{\mathrm{in}}^{-1}$. This state has zero first moment and covariance matrix
        \begin{equation}
            V(\ket{z_{\text{in}}^{-1}})=\bigoplus_{i=1}^m \begin{pmatrix}
                z_{\text{in}}^{-1} & 0 \\[4pt] 0 & z_{\text{in}} 
            \end{pmatrix}.
        \end{equation}
    As explained above, this step is experimentally feasible (see~\cref{sec:no-activesq}).
        
    \item Apply the unknown Gaussian unitary $G_{\mathbf r, S}$, so that the resulting state is $G_{\mathbf r, S} U_{\tilde{S}^{-1}}\ket{z_{\text{in}}^{-1}}^{\otimes m}$.
    \item Perform homodyne detection of the position quadrature.
\end{itemize}
The resulting state has
\begin{align}
    \mathbf{m} \bigl(G_{\mathbf r, S} U_{\tilde{S}^{-1}}\ket{z_{\mathrm{in}}^{-1}}^{\otimes m}\bigr) &= \mathbf{r}, \\ 
    V \bigl(G_{\mathbf r, S} U_{\tilde{S}^{-1}}\ket{z_{\mathrm{in}}^{-1}}^{\otimes m}\bigr) &= (\Delta+\id) V(\ket{z_{\mathrm{in}}^{-1}}) (\Delta+\id)^\top,
\end{align}
where $\Delta \coloneqq S\tilde{S}^{-1}-\id$. Let $\mathbf{r}_x$ be the position component of $\mathbf r$, and $A_{xx}$ the position–position block of a matrix $A$. 
Then homodyne detection of the position quadrature yields a random outcome $Y_x$ distributed as
\begin{equation}
    Y_x \sim \mathcal{N}\left(\mathbf{r}_x, \frac{V_{xx}}{2}\right),
\end{equation}
where
\begin{equation}
    V_{xx} \coloneqq \bigl( (\Delta+\id) V(\ket{z_{\mathrm{in}}^{-1}}) (\Delta+\id)^\top \bigr)_{xx}.
\end{equation}
Consequently, the covariance of $Y_x$ obeys
\begin{equation}
    \NORM{\frac{V_{xx}}{2}}_\infty 
    \le \frac{1}{2}\left(\frac{(1+\|\Delta\|_\infty)^2}{z_{\mathrm{in}}}+z_{\mathrm{in}}\|\Delta\|_\infty^2\right).
\end{equation}
Thus, in the regime $\|\Delta\|_\infty=o(1/z_{\text{in}})$, $V_{xx}$ vanishes as $z_{\text{in}} \to \infty$, enabling efficient estimation of $\mathbf r_x$.

\medskip

Let us now derive the query complexity of displacement learning using this approach.

\begin{lemma}[(Query complexity for learning the displacement without auxiliary-system entanglement)]\label{query_displ2}
Let $G_{\mathbf r, S} = D_{\mathbf r} U_S$ act on $m$ bosonic modes (phase-space dimension $2m$). Define $\Delta \coloneqq \tilde S^{-1}S - \id$ and fix a squeezing parameter $z_{\mathrm{in}}\ge 1$.  If $N_r\in\mathds{N}$ satisfies
\begin{equation}
N_r \ge \frac{2\bigl(\sqrt{2m}+\sqrt{2\log(2/\delta)}\bigr)^2}{\varepsilon^2}
\left(\frac{(1+\|\Delta\|_\infty)^2}{z_{\mathrm{in}}}+z_{\mathrm{in}}\|\Delta\|_\infty^2\right),
\end{equation}
then, using $2N_r$ queries to the unknown Gaussian unitary $G_{\mathbf r, S}$, we obtain an estimator $\tilde{\mathbf r}$ such that\begin{equation}
    \Pr\left[
   \|\tilde{\mathbf r} - \mathbf r\|_{2} \le \varepsilon
\right] \ge 1-\delta.    
\end{equation}
\end{lemma}
\begin{proof}
The proof is analogous to that of~\cref{thm:learn-r-nu-only}. 
Since $Y_p\sim \mathcal N(\mathbf r_p, V_{pp}/2)$, the empirical mean $\tilde{\mathbf r}_p$ satisfies the Gaussian concentration inequality (see~\cref{lem:estrelative}):
\begin{equation}
    \Pr\left[\norm{\tilde{\mathbf r}_p - \mathbf r_p}_2 \le \frac{\chi_{2m,\delta/2} \cdot \norm{{V_{pp}}/{2}}^{1/2}_\infty}{\sqrt{N_r}}\right] \ge 1-\frac{\delta}{2},
\end{equation}
where $\chi_{2m,\delta/2}=\sqrt{2m}+\sqrt{2\log(2/\delta)}$. From~\cref{eq_op_normvpp}, we have
\begin{equation}
    \NORM{\frac{V_{pp}}{2}}_\infty \le \frac12\left(\frac{(1+\|\Delta\|_\infty)^2}{z_{\mathrm{in}}}+z_{\mathrm{in}}\|\Delta\|_\infty^2\right).
\end{equation}
The same argument applies to the empirical mean $\tilde{\mathbf r}_x$. By a union bound, and recalling that $\tilde{\mathbf r}\coloneqq (\tilde{\mathbf r}_x,\tilde{\mathbf r}_p)$, we obtain
\begin{equation}
    \Pr\left[\norm{\tilde{\mathbf r}-\mathbf r}_2 \le 2\cdot\frac{\chi_{2m,\delta/2}}{\sqrt{N_r}} 
    \cdot \frac{1}{\sqrt{2}}\left(\frac{(1+\|\Delta\|_\infty)^2}{z_{\mathrm{in}}}+z_{\mathrm{in}}\|\Delta\|_\infty^2\right)^{1/2}\right] \ge 1-\delta.
\end{equation}
This completes the proof of the claimed bound.
\end{proof}

\section{End-to-end learning of Gaussian unitaries}\label{sec:end-to-end}
In the previous sections, we analyzed algorithms for estimating the symplectic component $S$ and the displacement vector $\mathbf r$ to a prescribed accuracy, along with their respective query complexities. We now turn to the central task of learning the full Gaussian unitary $G_{\mathbf r, S} = D_{\mathbf r} U_S$ to the desired accuracy. Importantly, the notion of accuracy must carry an operationally meaningful interpretation in the physical setting. In our work, this is captured by the \emph{energy-constrained diamond distance} (see~\cref{sec:energy-constrained-diamond-distance}), which provides a natural and operationally meaningful metric for quantifying the learning performance of Gaussian unitary channels, and underlies the problem definition stated in~\cref{problem:tomography}.

\subsection{Useful lemmas}
To prepare for the proof of the end-to-end learning guarantee for Gaussian unitaries, we first establish two key technical ingredients. The first is a transition from additive to multiplicative error bounds, which allows us to refine our accuracy guarantees. The second concerns fundamental properties of the regularized symplectic matrix, which will play a central role in controlling the operational error under the energy-constrained diamond distance.

\begin{lemma}[(From additive to multiplicative error)]
\label{lem:add-to-mult}
Let $S \in \mathrm{Sp}_{2m}(\R)$ and suppose $\tilde S$ satisfies 
$\norm{\tilde S - S}_{\infty} \le \varepsilon_S$ with $\varepsilon_S \norm{S}_{\infty} < 1$. 
Then
\begin{equation}
    \norm{\tilde S^{-1}S - \id}_{\infty} \le \frac{\varepsilon_S \norm{S}_{\infty}}{1-\varepsilon_S \norm{S}_{\infty}}.
\end{equation}
\end{lemma}
\begin{proof}
Write $\tilde S = S + P$ with $\norm{P}_{\infty} \le \varepsilon_S$. Then
\begin{align}
    \tilde S^{-1}S - \id &= (S+P)^{-1}S - \id \nonumber \\
    &= (\id + PS^{-1})^{-1} - \id.    
\end{align}
Setting $X = PS^{-1}$, note that $\norm{X}_{\infty} \le \varepsilon_S \norm{S^{-1}}_{\infty} < 1$.  
By the Neumann series,
\begin{equation}
    (\id+X)^{-\id} - \id = -X + X^2 - X^3 + \cdots,
\end{equation}
and termwise bounding yields
\begin{align}
    \norm{\tilde S^{-1}S - \id}_{\infty} 
    & \le \frac{\norm{X}_{\infty}}{1-\norm{X}_{\infty}} \nonumber \\
    & \le \frac{\varepsilon_S \norm{S^{-1}}_{\infty}}{1-\varepsilon_S \norm{S^{-1}}_{\infty}}.    
\end{align}
Since, by the Bloch--Messiah decomposition, $\norm{S^{-1}}_{\infty} = \norm{S}_{\infty}$ for symplectic matrices, the claim follows.
\end{proof}

\begin{lemma}[(Stability of the inverse for symplectic matrices)]
\label{lem:inv-perturb-symplectic}
Let $S \in \mathrm{Sp}_{2m}(\R)$, and set $z \coloneqq \norm{S}_{\infty}$.  
Suppose $\tilde S \in \R^{2m\times 2m}$ satisfies $\norm{\tilde S - S}_{\infty} \le \varepsilon_S$ with $ z \varepsilon_S< 1/2$. Then:
\begin{enumerate}
\item[(a)] $\tilde S$ is invertible.
\item[(b)] Its inverse satisfies $\norm{\tilde S^{-1}}_{\infty} \le 2z. $
\item[(c)] The multiplicative error is bounded by $\norm{\tilde S^{-1}S - \id}_{\infty} \le 2z\varepsilon_S.$
\item[(d)] Consequently, $\norm{\tilde S^{-1}S}_\infty \le 1 + 2z \varepsilon_S.$
\end{enumerate}
\end{lemma}

\begin{proof}
Since $S$ is symplectic, it is invertible and its singular values come in reciprocal pairs, hence 
$\norm{S}_{\infty}=\norm{S^{-1}}_{\infty}=z$.  
We can factor $\tilde S = S(\id+S^{-1}\bigl(\tilde S-S)\bigr)$. The assumption gives
\begin{align}
    \norm{S^{-1}(\tilde S-S)}_{\infty} 
    & \le \norm{S^{-1}}_{\infty} \cdot \norm{\tilde S-S}_{\infty} \nonumber \\
    & \le z\varepsilon_S < 1/2.
\end{align}
Thus $\id+S^{-1}(\tilde S-S)$ is invertible, with inverse given by the convergent Neumann series 
\begin{equation}
    \sum_{k=0}^\infty \bigl(-S^{-1}(\tilde S-S)\bigr)^k
\end{equation}
Hence $\tilde S$ is invertible and $\tilde S^{-1} = \bigl(\id + S^{-1}(\tilde S-S)\bigr)^{-1} S^{-1}$.    
Taking norms and using submultiplicativity,
\begin{align}
    \norm{\tilde S^{-1}}_{\infty} 
    & \le \norm{\bigl(\id + S^{-1}(\tilde S-S)\bigr)^{-1}}_{\infty}\cdot\norm{S^{-1}}_{\infty} \nonumber \\
    & \le \frac{z}{1-z\varepsilon_S} \le 2z,
\end{align}
which proves (b).  

\medskip
\noindent For (c),
\begin{align}
    \norm{\tilde S^{-1}S - \id}_{\infty} 
    & \le \norm{\tilde S^{-1}}_{\infty}\cdot\norm{\tilde S-S}_{\infty} \nonumber \\
    & \le 2z\varepsilon_S.
\end{align}

\medskip
\noindent Finally, (d) follows from
\begin{align}
    \norm{\tilde S^{-1}S}_\infty & \le 1 + \norm{\tilde S^{-1}S-\id}_\infty \nonumber \\
    & \le 1+2z\varepsilon_S.
\end{align}
This concludes all the proofs.
\end{proof}

\subsection{Query complexity of learning Gaussian unitaries}

Let us begin with the following preliminary result, which provides a way to propagate the error in the estimation of the displacement vector and the symplectic matrix to the error in estimating the Gaussian unitary with respect to the energy-constrained diamond norm.
\begin{lemma}[(From symplectic and displacement errors to diamond distance error)]\label{metalemma}
Let $G_{\mathbf r, S} = D_{\mathbf r} U_S$ be a Gaussian unitary on $m$ bosonic modes with symplectic matrix $S \in \mathrm{Sp}_{2m}(\mathds{R})$ satisfying $\|S\|_{\infty} \le z$.  
Let $\tilde{S} \in \mathrm{Sp}_{2m}(\mathds{R})$ and $\tilde{\mathbf r}\in\mathds{R}^{2m}$, and define
\begin{equation}
    \varepsilon_S \coloneqq \|\tilde{S} - S\|_\infty, \qquad 
    \varepsilon_r \coloneqq \|\tilde{\mathbf r} - \mathbf r\|_2.
\end{equation}
Then, the energy-constrained diamond distance between 
$\tilde{\mathcal G} \coloneqq \mathcal D_{\tilde{\mathbf r}} \circ \mathcal U_{\tilde S}$ 
and 
$\mathcal G \coloneqq \mathcal D_{\mathbf r} \circ \mathcal U_S$ satisfies
\begin{equation}
    \frac{1}{2}\|\tilde{\mathcal{G}} - \mathcal{G}\|_{\diamond,\bar n} 
    \le  12 \sqrt{9\sqrt{2m}(\bar n+1)} \sqrt{z\sqrt{2m}\,\varepsilon_S}
    + \sqrt{2}\sqrt{z^2 \bar n + 1} \, \varepsilon_r,
\end{equation}
where the diamond norm is taken with respect to the mean photon number constraint $\bar n$.  
In particular, for any $\varepsilon>0$, if it holds that
\begin{equation}
    \varepsilon_S \le \frac{\varepsilon^2}{2592 m z (\bar n+1)}, \qquad 
    \varepsilon_r \le \frac{\varepsilon}{2\sqrt{2}\sqrt{z^2 \bar n + 1}},
\end{equation}
then $\frac{1}{2}\|\tilde{\mathcal{G}} - \mathcal{G}\|_{\diamond,\bar n} \le \varepsilon$.
\end{lemma}

\begin{proof} The proof is divided in three steps.

\medskip
\noindent \emph{Step 1 (Symplectic contribution).}  Applying~\cref{thm:symplectic_bound} with $S_2=\tilde S$, we obtain
\begin{align}
\frac{1}{2} \norm{\mathcal U_{\tilde S} - \mathcal U_S}_{\diamond,\bar n} \le
\sqrt{\bigl(\sqrt{6}+\sqrt{10}+5\sqrt{2m}\bigr)(\bar n+1)}
\, g\bigl(\norm{\tilde S^{-1}S}_\infty\bigr)
\sqrt{\norm{\tilde S^{-1}S - \id}_2},
\end{align}
where $g(x)=\sqrt{{\pi}/{(x+1)}}+\sqrt{2x}$.  By~\cref{lem:inv-perturb-symplectic}, we have $\norm{\tilde S^{-1}}_{\infty} \le 2z$ and 
$\norm{\tilde S^{-1}S - \id}_{\infty} \le 2z\varepsilon_S$. Hence
\begin{align}
    \norm{\tilde S^{-1}S - \id}_2 & \le \sqrt{2m} \cdot\norm{\tilde S^{-1}S - \id}_{\infty} \nonumber \\
    & \le 2z\sqrt{2m}\,\varepsilon_S,
\end{align}
and moreover
\begin{equation}
    \norm{\tilde S^{-1}S}_\infty \le 1+2z\varepsilon_S.
\end{equation}
Since $z\varepsilon_S<1/2$ and $g(x)$ is increasing for $x\ge 0$, it follows that 
\begin{equation}
    g\bigl(\norm{\tilde S^{-1}S}_\infty\bigr)\le g(2)\le 4.
\end{equation}

Substituting these bounds gives
\begin{align}
\norm{\mathcal U_{\tilde S} - \mathcal U_S}_{\diamond,\bar n} & \le 2 \sqrt{\bigl(\sqrt{6}+\sqrt{10}+5\sqrt{2m}\bigr)(\bar n+1)} \, g\bigl(\norm{\tilde S^{-1}S}_\infty\bigr)\sqrt{\norm{\tilde S^{-1}S - \id}_2} \nonumber \\
& \le 2 \sqrt{\bigl(\sqrt{6}+\sqrt{10}+5\sqrt{2m}\bigr)(\bar n+1)} \cdot 4 \cdot \sqrt{2z\sqrt{2m}\,\varepsilon_S} \nonumber \\ 
& \le 12 \sqrt{\bigl(\sqrt{6}+\sqrt{10}+5\sqrt{2m}\bigr)(\bar n+1)} \sqrt{z\sqrt{2m}\,\varepsilon_S} \nonumber \\
& \le 12 \sqrt{9\sqrt{2m}(\bar n+1)} \sqrt{z\sqrt{2m}\,\varepsilon_S}. \label{eq:symplectic-bound-final}
\end{align}
where in the third line we used $8\sqrt{2}<12$, and in the last line we used 
$\sqrt{6}+\sqrt{10} < 4\sqrt{2} \le 4\sqrt{2m}$ since $m \in \N$. 

\medskip
\noindent
\emph{Step 2 (Displacement contribution).}  By~\cref{prop:displacement_bound},
\begin{align}
    \norm{\mathcal D_{\tilde{\mathbf r}} - \mathcal D_{\mathbf r}}_{\diamond,\bar n} &\le  2\sin\left(\min\left\{\frac{\bigl(\sqrt{\bar n}+\sqrt{\bar n+1}\bigr)}{\sqrt{2}}\varepsilon_r, \; \frac{\pi}{2}\right\}\right) \\\nonumber
    &\le  2\min\left\{\frac{\bigl(\sqrt{\bar n}+\sqrt{\bar n+1}\bigr)}{\sqrt{2}}\varepsilon_r, \; \frac{\pi}{2}\right\}\\\nonumber
    & \le \sqrt{2}\bigl(\sqrt{\bar n}+\sqrt{\bar n+1}\bigr)\varepsilon_r \\
    & \le 2\sqrt{2}\sqrt{\bar n+1}\,\varepsilon_r\,,
\end{align}
where in the second line, we use the elementary inequality $\sin x \le x$ for $x \ge 0$, and in the third line, we apply $\min\{a,b\} \le a$.

\medskip
\noindent
\emph{Step 3 (Putting together both contributions).}  
We have
\begin{align}
\frac12\norm{\tilde{\mathcal G} - \mathcal G }_{\diamond,\bar n}
& \le \frac12\norm{\mathcal{D}_{\tilde{\mathbf r}} \circ (\mathcal U_{\tilde S} - \mathcal U_S)}_{\diamond,\bar n}
+ \frac12\norm{(\mathcal D_{\tilde{\mathbf r}} - \mathcal D_{\mathbf r}) \circ \mathcal U_S}_{\diamond,\bar n} \nonumber \\[4pt]
& \le \frac12\norm{\mathcal U_{\tilde S} - \mathcal U_S}_{\diamond,\bar n}
+ \frac12\norm{\mathcal D_{\tilde{\mathbf r}} - \mathcal D_{\mathbf r}}_{\diamond, z^2 \bar n} \nonumber \\[4pt]
&\le \frac{\varepsilon}{2} + \sqrt{2}\sqrt{z^2\bar n+1}\,\varepsilon_r,
\end{align}

where:
\begin{itemize}
    \item the first inequality follows from the triangle inequality;
    \item the second inequality uses (i) invariance of the energy-constrained diamond norm under post-processing unitaries, i.e.
    \begin{equation}
        \frac12\norm{\mathcal{D}_{\tilde{\mathbf r}} \circ (\mathcal U_{\tilde S} - \mathcal U_S)}_{\diamond,\bar n}
        = \frac12\norm{\mathcal U_{\tilde S} - \mathcal U_S}_{\diamond,\bar n},    
    \end{equation}
    and (ii) that the mean energy of $U_S(\rho)$ is at most $\|S\|_\infty^2$ times the mean energy of $\rho$ (see e.g.~\cite{mele2024learning}); hence,
    \begin{equation}
        \frac12\norm{(\mathcal D_{\tilde{\mathbf r}} - \mathcal D_{\mathbf r}) \circ \mathcal U_S}_{\diamond,\bar n}
        \le \frac12\norm{\mathcal D_{\tilde{\mathbf r}} - \mathcal D_{\mathbf r}}_{\diamond, z^2 \bar n};
    \end{equation}
    \item the last inequality follows from Step 2 and Step 3.
\end{itemize}
This completes the derivation of the desired bound.
\end{proof}
To summarize, using~\cref{metalemma}, we can design an algorithm for the problem of learning Gaussian unitaries described in~\cref{problem:tomography}, carried out in two stages:
\begin{itemize}
    \item \emph{Learning the symplectic component with precision $\varepsilon_S \le \frac{\varepsilon^2}{2592 m z (\bar n+1)}$};
    \item \emph{Learning the displacement component with precision $\varepsilon_r \le \frac{\varepsilon}{2\sqrt{2}\sqrt{z^2 \bar n + 1}}$}.
\end{itemize}

The first stage can be carried out either with the vacuum-shared input protocol introduced in \cref{sec:vac-share}, or with the symmetric-probe protocol introduced in \cref{sec:sym-probe}. The second stage can be implemented either with the two-mode squeezed vacuum protocol introduced in \cref{sec:dis-entangle}, or with the single-mode squeezed protocol introduced in \cref{sec:algo-without-entanglement}.  

These options give rise to four distinct algorithms for learning Gaussian unitaries. With the tools developed above, one can easily establish upper bounds on the query complexity of each. In what follows, we focus on the algorithm that combines the vacuum-shared input protocol from \cref{sec:vac-share} with the two-mode squeezed vacuum protocol from \cref{sec:dis-entangle}, since this choice yields the best query complexity in the limit of large input probe energy $\bar{n}_{\text{in}}$. The overall procedure is summarized in~\cref{algo}, and its correctness is proved below.

\begin{algorithm}[H]
\caption[Learning Gaussian unitaries with auxiliary-system entanglement]{Learning Gaussian unitaries with auxiliary-system entanglement}
\label{algo}
\begin{algorithmic}[1]
\State \textbf{Input}: Setting of~\cref{problem:tomography}, with access to $(2m+1)N_S+N_r$ queries to the unknown Gaussian unitary $G_{\mathbf r, S}$, where $N_S$ and $N_r$ are specified in~\cref{thm:end-to-end-diamond}.
\State \textbf{Output}: $(\tilde{\mathbf{r}},\tilde{S})$ such that $D_{\tilde{\mathbf{r}}}U_{\tilde{S}}$ is $\varepsilon$-close to $G_{\mathbf r, S}$ in energy-constrained diamond norm with probability at least $ 1-\delta$.
\State Prepare $N_S$ copies of $G_{\mathbf r, S}\ket{0}$ and perform heterodyne detection on each of them to construct the mean estimator $\hat{Y}_0$.
\For{$i\in[2m]$}
\State Prepare $N_S$ copies of $G_{\mathbf r, S}\ket{\eta e_i}$ and perform heterodyne detection on each of them to construct the mean estimator $\hat{Y}_i$.
\EndFor
\State $\hat{S} \gets [\hat{Y}_1-\hat{Y}_0, \ldots, \hat{Y}_{2m}-\hat{Y}_0]$ 
\State $\tilde{S} \gets \hat{S} \cdot (-\Omega \hat{S}^{\top} \Omega \hat{S})^{-1/2}$
\State Prepare $N_r$ copies of $U_{S_\nu}^\dagger G_{\mathbf r, S} U_{\tilde{S}^{-1}} \ket{\nu}^{\otimes m}$ and perform heterodyne detection on each of them to construct the mean estimator $\tilde{\mathbf{m}}$ \emph{(see \cref{sec:no-activesq} for an experiment-friendly implementation using only passive optics and input squeezing)}.
\For{$i \in [2m]$}
    \State $\tilde{\mathbf{r}}_i \gets \tilde{\mathbf{m}}_i / \sqrt{\nu}$
\EndFor
\State \textbf{Return}: $\tilde{\mathbf{r}}$, $\tilde{S}$
\end{algorithmic}
\end{algorithm}

\begin{lemma}[(Learning valid symplectic and displacement via~\cref{algo})]\label{thm:end-to-end}
Let $G_{\mathbf r, S} = D_{\mathbf r} U_S$ be a Gaussian unitary on $m$ bosonic modes with displacement $\mathbf r \in \R^{2m}$ and symplectic matrix $S \in \mathrm{Sp}_{2m}(\R)$ satisfying $\norm{S}_{\infty} \le z$.  
Fix accuracy parameters $\varepsilon_S \in (0,1/(2z))$, $\varepsilon_r \in (0,1)$ and failure probability $\delta\in(0,1)$.  
There exists a protocol that outputs an estimate $(\tilde{\mathbf r}, \tilde S)$ such that
\begin{equation}
\Pr\left[
   \|\tilde S - S\|_{\infty} \le \varepsilon_S \;\;\mathrm{and}\;\;
   \|\tilde{\mathbf r} - \mathbf r\|_{2} \le \varepsilon_r
\right] \ge 1-\delta,
\end{equation}
where $\tilde S \in \mathrm{Sp}_{2m}(\R)$. The total query complexity is $(2m+1)N_S + N_r$, where
\begin{align}
    N_S & \ge \frac{324mz^6\bigl(\sqrt{2m}+\sqrt{2\log(2m/\delta)}\bigr)^{2}}{\eta^2 \varepsilon_S^2}, \label{eq:Ns-explicit}\\
    N_r &\ge \frac{\bigl(1+2\nu z \varepsilon_S + 6(\nu z \varepsilon_S)^2\bigr)\bigl(\sqrt{2m}+\sqrt{\log(2/\delta)}\bigr)^{2}}{\nu \varepsilon_r^2}, \label{eq:Nr-explicit}
\end{align}
with probe amplitude $\eta>0$ from~\cref{thm:learnS-again} and squeezing parameter $\nu\ge 1$ from~\cref{thm:learn-r-nu-only}.
\end{lemma}

\begin{proof}
By~\cref{prop:regularized-S-vac}, using $(2m+1)N_S$ queries suffices to obtain a symplectic $\tilde S$ with
\begin{equation}
\Pr\left[
   \|\tilde S - S\|_{\infty} \le \varepsilon_S
\right] \ge 1-\frac{\delta}{2}.    
\end{equation}
This directly implies the requirement in~\cref{eq:Ns-explicit}.

\medskip
Set $\Delta = \tilde S^{-1}S - \id$.  
By~\cref{lem:add-to-mult}, 
\begin{equation}
    \norm{\Delta}_\infty \le \frac{ z \varepsilon_S}{1- z \varepsilon_S} \le 2 z \varepsilon_S,
\end{equation}
where the last inequality uses $\varepsilon_S < 1/(2z)$.  
Substituting this bound into the requirement of~\cref{thm:learn-r-nu-only} yields~\cref{eq:Nr-explicit}.

\medskip
Each stage succeeds with probability at least $1-\delta/2$.  
By the union bound, the joint estimate $(\tilde{\mathbf r}, \tilde S)$ satisfies both error guarantees with probability at least $1-\delta$.  
The total query complexity is $(2m+1)N_S+N_r$.
\end{proof}
We now combine the results obtained so far to derive the final query complexity guarantee for learning Gaussian unitary channels under the energy-constrained diamond norm, via the procedure of~\cref{algo}.
 
\begin{theorem}[(Learning Gaussian unitaries in energy-constrained diamond norm via~\cref{algo})]
\label{thm:end-to-end-diamond}
Let $m\in\mathds{N}$, $z\ge 1$, $\bar n>0$, $\bar n_{\mathrm{in}}>0$, 
$\varepsilon\in(0,1)$, and $\delta\in(0,1)$ be known parameters as in~\cref{problem:tomography}.  Let $\eta>0$, $\varepsilon_S>0$, $\varepsilon_r>0$ and $\nu\ge 1$ such that
\begin{equation}
    \eta \le \sqrt{\bar{n}_{\mathrm{in}}}, \quad
    \nu \le 1 + \frac{\bar{n}_{\mathrm{in}}}{2m}, \quad
    \varepsilon_S \le \frac{\varepsilon^2}{2592 m z (\bar n+1)}, \quad
    \varepsilon_r \le \frac{\varepsilon}{2\sqrt{2}\sqrt{z^2\bar n+1}}, \label{eq:algo1-parameters}
\end{equation}
and let $N_S,N_r\in\mathds{N}$ satisfying
\begin{align}
    N_S & \ge \frac{324mz^6\bigl(\sqrt{2m}+\sqrt{2\log(2m/\delta)}\bigr)^{2}}{\eta^2 \varepsilon_S^2}, \label{eq:final-Ns}\\
    N_r &\ge \frac{\bigl(1+2\nu z \varepsilon_S + 6(\nu z \varepsilon_S)^2\bigr)\bigl(\sqrt{2m}+\sqrt{\log(2/\delta)}\bigr)^{2}}{\nu \varepsilon_r^2}. \label{eq:Final-Nr}
\end{align}
Then~\cref{algo} is such that:
\begin{itemize}
    \item \textbf{\emph{Given:}} black-box access to an unknown $m$-mode Gaussian unitary $G_{\mathbf r,S} = D_{\mathbf r} U_S$, where $D_{\mathbf r}$ is the displacement operator and $U_S$ is a symplectic Gaussian unitary specified by a symplectic matrix $S$ with operator norm 
    $\|S\|_\infty \le z$;
    \item \textbf{\emph{Using:}} a number of queries to $G_{\mathbf r,S}$ of
    \begin{equation}
        N_\textnormal{tot}\coloneqq(2m+1)N_S + N_r
    \end{equation}
    and only input states with mean photon number of at most $\bar n_{\mathrm{in}}$;
    \item \textbf{\emph{Outputs:}} estimators $\tilde{\mathbf r} \in \R^{2m}$ and $\tilde S \in \mathrm{Sp}_{2m}(\R)$; the corresponding Gaussian unitary channel $\tilde{\mathcal G} \coloneqq \mathcal D_{\tilde{\mathbf r}} \circ \mathcal U_{\tilde S}$ approximates the true channel $\mathcal G = \mathcal D_{\mathbf r} \circ \mathcal U_S$ and satisfies
    \begin{equation}
        \Pr\left[
            \frac{1}{2}\norm{\tilde{\mathcal{G}} - \mathcal{G}}_{\diamond,\bar n} 
            \le \varepsilon 
        \right] \ge 1-\delta,
    \end{equation}
    where the diamond norm is taken with respect to the mean photon number constraint $\bar n$.
\end{itemize}
\end{theorem}

\begin{proof}
By~\cref{thm:end-to-end}, the algorithm outputs $(\tilde{\mathbf r}, \tilde S)$, with $\tilde S$ symplectic, such that
\begin{equation}\label{eq_00000}
    \Pr\left[\norm{\tilde{\mathbf r}-\mathbf r}_2 \le \varepsilon_r \;\; \text{and} \;\; \norm{\tilde S-S}_{\infty} \le \varepsilon_S\right] \ge 1-\delta,
\end{equation}
while using 
\[
    N_\textnormal{tot} = (2m+1)N_S + N_r
\]
queries of the unknown Gaussian unitary, where $N_S$ and $N_r$ are defined in~\cref{eq:final-Ns,eq:Final-Nr}, respectively.  
Then, by~\cref{metalemma}, the guarantee in~\cref{eq_00000}, together with the conditions 
\[
    \varepsilon_S \le \frac{\varepsilon^2}{2592 m z (\bar n+1)} 
    \quad \text{and} \quad 
    \varepsilon_r \le \frac{\varepsilon}{2\sqrt{2}\sqrt{z^2\bar n+1}},
\]
implies directly that
\[
    \Pr\left[
        \frac{1}{2}\norm{\tilde{\mathcal{G}} - \mathcal{G}}_{\diamond,\bar n} 
        \le \varepsilon 
    \right] \ge 1-\delta.
\]

It remains to discuss the constraint on the mean photon number $\bar{n}_{\text{in}}$ of the input states employed in~\cref{algo}.  
As described there, estimating the symplectic component requires coherent states of amplitude $\eta$, each with mean photon number $\eta^{2}$ (see~\cref{eq_mean_photon}).  
On the other hand, estimating the displacement component relies on the state $\ket{\nu}^{\otimes m}$, namely the tensor product of $m$ two-mode squeezed vacuum states, each with covariance matrix given in~\cref{eq_tmsv}. By~\cref{eq_mean_photon}, the mean photon number of $\ket{\nu}^{\otimes m}$ is $2m(\nu-1)$.  Therefore, we must enforce the constraints 
\[
    \eta^{2} \leq \bar{n}_{\mathrm{in}}
    \quad \text{and} \quad 
    2m(\nu-1) \leq \bar{n}_{\mathrm{in}}.
\]
Equivalently, this yields 
\[
    \eta \le \sqrt{\bar{n}_{\mathrm{in}}}
    \quad \text{and} \quad 
    \nu \le 1 + \frac{\bar{n}_{\mathrm{in}}}{2m},
\]
which completes the proof.
\end{proof}

We now specialize the general query complexity bound of 
\cref{thm:end-to-end-diamond} to a concrete choice of parameters, and highlight its behavior in the limit of infinite input energy.
\begin{remark}[(Query complexity of~\cref{algo})]
\label{remark:scaling-probepower-nu}
Let $m\in\mathds{N}$, $z\ge 1$, $\bar n>0$, $\bar n_{\mathrm{in}}>0$, $\varepsilon\in(0,1)$, and $\delta\in(0,1)$ be known parameters as in~\cref{problem:tomography}. Assume that $\bar{n}_{\mathrm{in}}\ge (2m)^{4/3}$. Applying~\cref{thm:end-to-end-diamond} with the parameter choices  
\begin{equation}
    \eta \coloneqq \sqrt{\bar{n}_{\mathrm{in}}}, 
    \quad \nu \coloneqq \bar{n}_{\mathrm{in}}^{1/4}+1, 
    \quad \varepsilon_S \coloneqq\frac{\varepsilon^2}{2592mz(\bar n+1)(\bar{n}_{\mathrm{in}} +1)^{1/4}}, \quad \varepsilon_r \coloneqq \frac{\varepsilon}{2\sqrt{2}\sqrt{z^2\bar n+1}},
\end{equation}
we find that~\cref{algo} solves~\cref{problem:tomography} with a total query complexity $N_\mathrm{tot} = (2m+1)N_S + N_r$, where $N_S,N_r \ge 1$ are integers satisfying
\begin{align}
    N_S & = \Theta\left(\frac{m^3z^8(\bar{n}+1)^2{(\bar{n}_{\mathrm{in}}+1)^{1/2}}\bigl(\sqrt{m}+\sqrt{\log(m/\delta)}\bigr)^2}{\bar{n}_{\mathrm{in}}\varepsilon^4}\right), \label{eq:N_S-sample}\\
    N_r & = \Theta\left(\frac{(z^2\bar{n}+1)\bigl(\sqrt{m}+\sqrt{\log(1/\delta)}\bigr)^2}{(1+\bar{n}_{\mathrm{in}}^{1/4})\varepsilon^2}\right). \label{eq:N_r-sample}  
\end{align}
Furthermore, in the limit of infinite input energy, $\bar{n}_{\mathrm{in}} \to \infty$, both $N_S$ and $N_r$ converge to $1$, and hence the total query complexity of~\cref{algo} reduces to $N_\mathrm{tot} = 2m + 2$.
\end{remark}

Finally, we record a remark about the overall time efficiency of~\cref{algo}.
\begin{remark}[(Time efficiency)]\label{rem:time-efficiency}
The measurement steps and the construction of $\tilde{\mathbf{r}}$ require only $\mathcal{O}(m)$ time, 
while each matrix multiplication can be performed in $\mathcal{O}(m^\omega)$ time, 
where $\omega$ denotes the matrix multiplication exponent. 
The most expensive subroutine is the computation of the principal matrix root, 
which can be carried out in $\mathcal{O}(m^{3})$ time using, for instance, the Schur decomposition~\cite{higham1987computing}. This computation dominates the symplectic regularization procedure, as the other steps are inversion and matrix multiplication, which each take less than $O(m^3)$ time.
Therefore, the time complexity of the algorithm is $O(m^3)$ plus the query complexity. In any case, we obtain a time complexity that is $\mathsf{poly}(m, \squeezingparameter,\bar n,\bar{n}_\mathrm{in},1/\varepsilon)$.
\end{remark}

\addcontentsline{toc}{section}{Acknowledgements}
\section*{Acknowledgements}
We thank Scott Aaronson, Mark Abate, Lennart Bittel, Jens Eisert, Daniel Stilck França, Michael Jaber, Hyun-Soo Kim, Ludovico Lami, Lorenzo Leone, Salvatore F.~E.~Oliviero, Jacopo Rizzo, Cambyse Rouzé, and Chirag Wadhwa for useful conversations.

V.~I. is funded by an NSF Graduate Research Fellowship. J.~L. is supported by the National Research Foundation of Korea (NRF) through grants funded by the Ministry of Science and ICT (Grant No.~RS-2024-00404854). A.~A.~M. acknowledges financial support from BMBF (FermiQP, MuniQC-Atoms, DAQC), BMWK (EniQmA), the Quantum Flagship (Millenion, PasQuans2). F.~A.~M. acknowledges financial support from the European Union (ERC StG ETQO, Grant Agreement no.\ 101165230). Views and opinions expressed are however those of the authors only and do not necessarily reflect those of the European Union or the European Research Council. Neither the European Union nor the granting authority can be held responsible for them.  F.~A.~M.~acknowledges financial support from the project: PRIN 2022 ``Recovering Information in Sloppy QUantum modEls (RISQUE)'', code 2022T25TR3, CUP E53D23002400006. F.~A.~M.~thanks California Institute of Technology for hospitality.

\addcontentsline{toc}{section}{References}
\bibliographystyle{unsrt}
\bibliography{biblio}

\newpage
\appendix

\section{Proof of~\cref{lem:moments}}\label{app:proof-moments}

Starting from $|\nu\rangle^{\otimes m}$ (produced by $S_\nu$ acting on vacuum),
we apply $U_{\tilde S^{-1}}$, then $G_{\mathbf r, S}=D_{\mathbf r}U_S$, and finally $U_{S_\nu}^\dagger$. Ignoring the displacement for the moment, the net symplectic acting in the Schr\"odinger picture is
\begin{equation}
    W = S_{\nu}^{-1}
    \begin{pmatrix}
        S & 0\\ 0 & \id
    \end{pmatrix}
    \begin{pmatrix}
        \tilde S^{-1} & 0\\ 0 & \id
    \end{pmatrix}
    S_{\nu}.    
\end{equation}
It is convenient to set $\Delta \coloneqq S\tilde S^{-1} - \id$ so that $S\tilde S^{-1} = \id + \Delta$. A direct block multiplication using the explicit form of $S_{\pm\nu}$ gives
\begin{align}
W &= S_{\nu}^{-1}
    \begin{pmatrix}
        \id + \Delta & 0\\ 0 & \id
    \end{pmatrix}
    S_{\nu}
=
\begin{pmatrix}
    W_{11} & W_{12}\\
    W_{21} & W_{22}
\end{pmatrix},\\
W_{11} &= \nu(\id+\Delta) - (\nu-1)\id = \id + \nu\Delta,\\
W_{12} &= \sqrt{\nu(\nu-1)}\big((\id+\Delta)Z - Z\big) = \sqrt{\nu(\nu-1)}\Delta Z,\\
W_{21} &= \sqrt{\nu(\nu-1)}\big(-Z(\id+\Delta) + Z\big) = -\sqrt{\nu(\nu-1)}Z\Delta,\\
W_{22} &= -(\nu-1)Z(\id+\Delta)Z + \nu \id = \id - (\nu-1)Z\Delta Z.
\end{align}
Here we used $Z^\top=Z$ and $Z^2=\id$ repeatedly.

\paragraph{First moments.}
Including the displacement from $G_{\mathbf r, S}=D_{\mathbf r}U_S$,
the output mean vector is 
\begin{equation}
    \mathbf{m} \bigl(U_{S_\nu}^\dagger G_{\mathbf r, S}U_{\tilde{S}^{-1}}\ket{\nu}^{\otimes m}\bigr)
    = S_{\nu}^{-1}\,(\mathbf r,\mathbf 0)
    = \big(\sqrt{\nu}\,\mathbf r, -\sqrt{\nu-1}Z\mathbf r\big)    
\end{equation}

\paragraph{Second moments.}
Since $|\nu\rangle^{\otimes m}$ is obtained from vacuum by $U_{S_\nu}$,
we can write the output covariance as 
\begin{equation}
    V\bigl(U_{S_\nu}^\dagger G_{\mathbf r, S}U_{\tilde{S}^{-1}}\ket{\nu}^{\otimes m}\bigr) = WW^\top=\begin{pmatrix} A & C\\[2pt] C^\top & B\end{pmatrix}.
\end{equation} 
Using the block form above and $Z^\top=Z$, $Z^2=\id$, we obtain
\begin{align}
A & = W_{11}W_{11}^\top + W_{12}W_{12}^\top \nonumber \\
& = (\id+\nu\Delta)(\id+\nu\Delta)^\top + \nu(\nu-1)\Delta\Delta^\top
\nonumber\\
&= \id + \nu(\Delta+\Delta^\top) + \nu(2\nu-1)\Delta\Delta^\top,
\\[4pt]
B &= W_{21}W_{21}^\top + W_{22}W_{22}^\top \nonumber\\
&= \nu(\nu-1)\,Z\Delta\Delta^\top Z + \id - (\nu-1)Z(\Delta+\Delta^\top)Z + (\nu-1)^2Z\Delta\Delta^\top Z
\nonumber\\
&= \id - (\nu-1)Z(\Delta+\Delta^\top)Z + (\nu-1)(2\nu-1)Z\Delta\Delta^\top Z, \nonumber \\
&= \id - (\nu-1)(\Delta+\Delta^\top) + (\nu-1)(2\nu-1)\Delta\Delta^\top.
\\[4pt]
C &= W_{11}W_{21}^\top + W_{12}W_{22}^\top \nonumber\\
&= (\id+\nu\Delta)\big(-\sqrt{\nu(\nu-1)}\Delta^\top Z\big) + \sqrt{\nu(\nu-1)}\Delta Z\big(\id-(\nu-1)Z\Delta Z\big)^\top
\nonumber\\
&= \sqrt{\nu(\nu-1)}\Big[\Delta Z - \Delta^\top Z - (2\nu-1)\Delta\Delta^\top Z\Big]
\nonumber\\
&= \Big[-(2\nu-1)\sqrt{\nu(\nu-1)}\Delta\Delta^\top - \sqrt{\nu(\nu-1)}\Delta^\top + \sqrt{\nu(\nu-1)}\Delta\Big]Z.
\end{align}

\end{document}